\documentclass[10pt,english,british]{article}
\usepackage[T1]{fontenc}
\usepackage[latin9]{inputenc}
\usepackage[paperwidth=8.5in,paperheight=11in]{geometry}
\usepackage{amsthm}
\usepackage{amsmath}
\usepackage{amssymb}
\usepackage{graphicx}
\usepackage{setspace}
\usepackage{esint}
\usepackage[authoryear]{natbib}
\usepackage{titlesec}
\usepackage[margin=0.2cm]{caption}
\makeatletter

\titleformat{\section}
  {\normalfont\Large\bfseries\centering}{\thesection}{1em}{}

\newcommand{\lyxaddress}[1]{
\par {\raggedright #1
\vspace{1.4em}
\noindent\par}
}

\@ifundefined{showcaptionsetup}{}{%
 \PassOptionsToPackage{caption=false}{subfig}}
\usepackage{subfig}
\makeatother

\theoremstyle{plain}
\newtheorem{thm}{\protect\theoremname}

\makeatother

\usepackage{babel}
\addto\captionsbritish{\renewcommand{\theoremname}{Theorem}}
\addto\captionsenglish{\renewcommand{\theoremname}{Theorem}}
\providecommand{\theoremname}{Theorem}

\begin{document}

\title{Bayesian Parameter Estimation for Latent Markov Random Fields and
Social Networks}

\author{Richard G. Everitt\footnote{Richard G. Everitt is a Postdoctoral Scientist at the University of Oxford, 1 South Parks Road, Oxford OX1 3TG, UK (e-mail: richard.g.everitt@gmail.com).}}

\maketitle

\lyxaddress{Nuffield Department of Clinical Medicine,\\
University of Oxford,\\
John Radcliffe Hospital,\\
Oxford OX3 9DU, UK.\\
Tel.: +447815763118\\
\texttt{richard.g.everitt@gmail.com}
}

\pagebreak

\begin{abstract}
Undirected graphical models are widely used in statistics, physics
and machine vision. However Bayesian parameter estimation for undirected models
is extremely challenging, since evaluation of the posterior typically
involves the calculation of an intractable normalising constant. This
problem has received much attention, but very little
of this has focussed on the important practical case where the data
consists of noisy or incomplete observations of the underlying hidden
structure. This paper specifically addresses this problem, comparing
two alternative methodologies. In the first of these approaches particle
Markov chain Monte Carlo \citep{Andrieu2010} is used to efficiently
explore the parameter space, combined with the exchange
algorithm \citep{Murray2006} for avoiding the calculation of the
intractable normalising constant (a proof showing that this combination targets the correct distribution in found in a supplementary appendix online). This approach is compared with approximate
Bayesian computation \citep{Pritchard1999}. Applications to estimating
the parameters of Ising models and exponential random graphs from noisy data are presented.  Each algorithm used in the paper targets an approximation to the true posterior due to the use of MCMC to simulate from the latent graphical model, in lieu of being able to do this exactly in general. The supplementary appendix also describes the nature of the resulting approximation.
\end{abstract}


\section{INTRODUCTION}

\subsection{Motivation}

The subject of this paper is Bayesian inference in models for large
numbers of dependant objects, e.g. pixels in an image, people in a
social network, or pages on the world wide web. We focus on \emph{Markov
random fields} (MRFs), also known as \emph{undirected graphical models}
(UGMs), which are the models most commonly used for this type of data.
Much of the literature (e.g. \citet{Murray2006,Moller2006}) is devoted
to estimating the parameters of these models in the case where the
MRF is completely observed, but for real data this is often not the
case: in practice observations of the MRF can be noisy or incomplete.
As a result, estimation of the parameters of the model presents a
significant computational challenge, particularly when using Bayesian
estimation in order to account for the missing data in a principled
manner.

Let $y\in\mathcal{Y}$ represent noisy or incomplete observations
of the hidden variables $x\in\mathcal{X}$. Our aim is to estimate
the parameters $\theta\in\Theta$ of a model for this data. Using
the Bayesian approach, we describe a joint distribution $p(\theta,x,y)=p(\theta)f(x,y|\theta)$
over $\theta$, $x$ and $y$ and wish to estimate $\theta$ through
the posterior distribution:
\begin{equation}
p(\theta|y)\propto\int_{x}p(\theta)f(x,y|\theta)dx.\label{eq:posterior}
\end{equation}

Our focus is on Monte Carlo methods for simulating from $p(\theta|y)$,
and we present two quite general methodologies for achieving this.
We apply these methods to two well known types of data: a noisy image
(similar to \citet{JulianBesag1986}); and an indirectly observed
social network (similar to \citet{Caimo2011}).

Parameter estimation in problems such as these is often addressed
using maximum likelihood (as in \citet{Heckerman1996}), rather than
a fully Bayesian approach (aside from the recent work in \foreignlanguage{english}{\citet{Moller2004,Murray2006,Friel2009,Koskinen2010}}).
In this paper we compare the computational approach taken in these
recent papers, which we show is not always appropriate, to alternative
computational methods.

In some cases we are also interested in
the posterior distribution over the hidden variables $x$ (in the
applications described above, this is the case if we wish to infer
a denoised Ising model or social network) but this is not our primary
focus. A part of the methodology we describe is the use of sequential Monte Carlo (SMC) samplers
for simulating from the posterior of $x$, which can be significantly
more efficient than using MCMC on these models.

In the remainder of this section we describe MRFs in more detail in
section \ref{sub:Models}, and basic methods of inference in these
model in section \ref{sub:Inference}. Section \ref{sub:Inference}
also describes the challenges faced in using these basic inference
methods, and this motivates the rest of the paper.

\selectlanguage{english}%

\subsection{Models\label{sub:Models}}

\selectlanguage{british}%

\subsubsection{Markov Random Fields}

\selectlanguage{english}%
Throughout the paper, we take $p(\theta)$ to have some simple form
that can be both evaluated pointwise and simulated from using standard
techniques. The primary purpose of this paper is to discuss computational
methods for inference rather than choosing the most appropriate models,
therefore we choose a proper prior simply to ensure that the posterior
always exists.

We assume that the likelihood $f(x,y|\theta)$ \foreignlanguage{british}{factorises}
as an M\foreignlanguage{british}{RF}, i.e. the joint distribution
is completely defined by \emph{potential} functions over groups of
variables known as \emph{cliques}\foreignlanguage{british}{ \citep{Besag1974}:}
\begin{equation}
f(x,y|\theta)=\frac{1}{Z(\theta)}\prod_{j=1}^{M}\Phi_{j}(x,y\in C_{j}|\theta)
\end{equation}
where $\Phi_{1:M}$ are \emph{potentials}
on cliques $C_{1:M}$, and the normalising constant
\begin{equation}
Z(\theta)=\int_{x,y}\prod_{j=1}^{M}\Phi_{j}(x,y\in C_{j}|\theta)dxdy\label{eq:Z}
\end{equation}
is usually an intractable integral.
The methods described in this paper also apply to directed
graphical models, which is easier to deal with since the factors of the joint distribution in a directed model are simply conditional distributions that can be normalised independently \citep{Heckerman1996}, and hence 
the intractable normalising constant is not present.

\subsubsection{Ising Models\label{sub:Ising-models}}

The Ising model is a particular MRF that defines a distribution over
a random vector $x$ of binary variables each of which can take the
value $-1$ or $1$. The distribution is defined by the pairwise interactions
between variables as $f(x|\theta_{x})=\exp\left[\theta_{x}\sum_{(i,j)\in\mathbf{N}}x_{i}x_{j}\right]/Z(\theta_{x})$, where $\mathbf{N}$ is a set that defines pairs of nodes that are
{}``neighbours'' and, for our purposes, $\theta_{x}>0$. This distribution
follows through choosing cliques containing only two nodes, with the
potential function \foreignlanguage{english}{$\Phi(x_{i},x_{j}|\theta_x)=\exp(\theta_x x_{i}x_{j})$
for each clique.} In this paper we concentrate on the case
where $\mathbf{N}$ is chosen so that the variables are arranged in
a two-dimensional grid. The study of these models often centres around
the phase change that the model exhibits when $\theta_{x}$ passes
a critical value: a small $\theta_{x}$
gives a small preference for neighbouring nodes to be the same and
results in a high probability that localised groups of variables take
on the same value. When $\theta_{x}$ increases above the critical
value the joint distribution has the vast majority of its probability
mass on the cases where almost every variable takes on the same value.

The Ising model and its generalisation, the Potts model (where the variables can take on more than two states), \foreignlanguage{english}{are
frequently used in analysing spatially structured data, especially
images. }In these applications it is usually the case that the random
vector $x$ is observed indirectly through observations $y$. To model
the noisy image data used later in the paper, we assume that each
$x$ variable has a corresponding $y$ variable representing a noisy
observation of the $x$ variable. We then take the joint over the
$x$ and $y$ variables to be: \foreignlanguage{english}{
\begin{equation}
f(x,y|\theta)=f(x|\theta_{x})g(y|x,\theta_{y})=\frac{1}{Z(\theta_{x})}\exp\left[\theta_{x}\sum_{(i,j)\in\mathbf{N}}x_{i}x_{j}\right]\frac{1}{Z(\theta_{y})}\exp\left[\theta_{y}\sum_{k=1}^{T}x_{k}y_{k}\right],\label{eq:general_gm}
\end{equation}
where $T$ is the number of pixels in the image, $Z(\theta_{y})=\left(\exp(\theta_{y})+\exp(-\theta_{y})\right)^{T}$
is the normalising constant for all of the potentials on the $(x_{k},y_{k})$
pairs.}

\subsubsection{Exponential Random Graphs}

This section considers the most popular model for social networks:
the \emph{exponential random graph model} (ERGM) or $p^{*}$ model
\citep{Wasserman1996}. These models are also used in physics and
biology and aim to model data consisting of nodes and edges, which
in the social network context represent \emph{actors} and \emph{relationships}
between these actors. There has been relatively little work on using
Bayesian inference for inferring the parameters of ERGMs aside from recent
papers by \citet{Koskinen2010,Caimo2011}.

For an ERGM, the random vector $x$ is defined over the space of all
graphs on a set of nodes, with each variable in $x$ representing
the presence or absence of a particular edge ($x_{k}$ takes value
1 when edge $k$ is present and 0 when edge $k$ is absent). An ERGM
for a graph (which consists of a set of edges over the set of nodes)
is then given by:
\begin{equation}
f(x|\theta_{x})=\frac{1}{Z(\theta_{x})}\exp(\theta_{x}^{T}S(x))\label{eq:ergm}
\end{equation}
where $S(x)$ is a vector of statistics of the
graph (e.g. the number of edges, triangles, etc) and $\theta_{x}^{T}$ denotes the transpose of the parameter vector $\theta_{x}$. The normalising
constant is intractable due to the extremely large number of possible
graphs. ERGMs are simply MRFs defined on the edge space of networks
\citep{Frank1986}, with statistics of the network corresponding to
particular cliques in the MRF.

We consider the case where the $y$ variables
represent noisy observations about the existence of
each edge in the network. Analogous to equation (\ref{eq:general_gm})
we use the model
\begin{equation}
f(x,y|\theta)=f(x|\theta_{x})g(y|x,\theta_{y})=\frac{1}{Z(\theta_{x})}\exp(\theta_{x}^{T}S(x))\frac{1}{Z(\theta_{y})}\exp\left[\theta_{y}\sum_{k=1}^{T}(2x_{k}-1)(2y_{k}-1)\right].\label{eq:ergm_noise}
\end{equation}
We note that the algorithms described in this paper are also applicable
to the alternative missing data scenarios described in \citet{Koskinen2010}.  In fact, the only specific detail that is required to use the methods in this paper is the factorisation of $x$ as an MRF: the space in which $\theta$ and $y$ lie affects only minor implementation details.

\selectlanguage{english}%

\subsection{Inference\label{sub:Inference}}

\subsubsection{Posterior Inference\label{sub:Posterior-inference}}

Recall that the posterior distribution of interest is $p(\theta|y)$.
We now consider standard Monte Carlo methods for sampling from this
distribution, under MRF models such as those described
in the previous section. Throughout the algorithmic sections of this
paper we consider the general case where both the observed variables
$y$ and unobserved variables $x$ are jointly distributed according
to an MRF, with distribution $f(x,y|\theta)$.

\selectlanguage{british}%
The high dimensionality of $x$ leads to the use of Markov chain Monte
Carlo (MCMC) for performing the integration in equation (\ref{eq:posterior}).
Specifically, MCMC is used to simulate from $p(\theta,x|y)$, projecting
the obtained points into the marginal space in order to obtain a sample
from $p(\theta|y)$. In these circumstances, since it is unclear as
to how to propose an efficient move in the joint space of $\theta$
and $x$, the standard approach is to use an MCMC algorithm that updates
$\theta$ and $x$ separately through moves on their full conditional
distributions (see \citet{Murray2004,Friel2009,Koskinen2010}, for
example). We shall refer to this method as the {}``data augmentation''
(DA) approach.

Superficially the DA approach may seem to be extremely simple, but
for the models described here there are three major difficulties with
using this method:
\selectlanguage{english}%
\begin{enumerate}
\item \textbf{Sampling from }$p(\theta|x,y)$\textbf{ is difficult,} since
this requires the evaluation of the \emph{intractable normalising
constant} in equation (\ref{eq:Z}).
\item \textbf{Sampling from }$p(x|\theta,y)$\textbf{ can be difficult,}
since this is usually a distribution with a complicated structure
on a high-dimensional space.
\item \textbf{Using a data augmentation approach may be inefficient, }since
$x$ and $\theta$ are often highly dependent \emph{a posteriori}.
\end{enumerate}
\selectlanguage{british}%
These difficulties largely dictate the structure of the remainder
of the paper. In section \ref{sec:Exchange-marginal-PMCMC} we devise a novel MCMC algorithm for sampling from $p(\theta|y)$
that is designed to minimise the effect of each of these problems,
combining the \emph{exchange algorithm} of \citet{Murray2006} with
\emph{marginal particle Markov chain Monte Carlo} (PMCMC) \citep{Andrieu2010}.
In section \ref{sec:Approximate-Bayesian-computation} explore an alternative approach to sampling from $p(\theta|y)$,
using \emph{approximate Bayesian computation} (ABC) \citep{Pritchard1999}
to largely circumvent the difficulties listed above. However, this
approach introduces complications of its own. Section \ref{sec:Applications} applies both approaches to data, with a comparison to the standard DA approach.

Throughout the paper we make use of a single site Gibbs sampler on
$p(x|\theta,y)$, and for this reason we devote some space in the
remainder of this section to describing this algorithm.

\selectlanguage{english}%

\subsubsection{Gibbs Samplers for MRFs\label{sub:Gibbs-samplers-for}}

\selectlanguage{british}%
The simplest approach to sampling from $p(x|\theta,y)$ for a MRF model using MCMC is to use single component Metropolis, where the variables
are updated one at a time. Variable $k$ is updated using a Metropolis-Hastings
step targeting the full conditional $p(x_{k}|x_{\neq k},\theta,y)\propto\prod_{C_{j}\ni x_{k}}\Phi_{j}(x,y\in C_{j}|\theta)$, where the product is over the potentials of all cliques that contain
variable $k$. When there are strong dependencies between the variables
this results in an inefficient MCMC algorithm. A partial solution
might be to update highly correlated variables in blocks: \citet{Andrieu2010}
note that to create an efficient algorithm the size of the blocks
of variables must be limited, since it becomes more increasingly difficult
to design good proposals for a block as the size of the block grows.
The inefficiency of Gibbs samplers for Ising models is well known
(e.g. \citet{Higdon1998}) and a similar observation is made in the literature
on ERGMs \citep{Snijders2002}. Similar approaches, with the same drawbacks, are commonly used for sampling from $p(x|\theta)$, which is encountered both in the exchange and ABC algorithms later in the paper. The inefficiency of these approaches motivates the use
of SMC samplers, where we observe
a significant improvement over the Gibbs sampler.

\section{EXCHANGE MARGINAL PMCMC\label{sec:Exchange-marginal-PMCMC}}

This section describes the design of an MCMC algorithm for simulating
from $p(\theta|y)$ that addresses the problems listed in section
\ref{sub:Posterior-inference}. We begin in section \ref{sub:Intractable-normalising-constant}
by describing methods for sampling from $p(\theta|x,y)$ that avoid
the evaluation of the intractable term $Z(\theta)$ in equation
(\ref{eq:Z}). Our favoured approach is the exchange algorithm of \citet{Murray2006}:
an algorithm that requires exact simulation from $f(x,y|\theta)$
which is not generally possible in the models we consider. We study
the effect of replacing exact simulation with the use of MCMC.

The remainder of the section focuses on the use of marginal PMCMC
to sample efficiently in cases where $\theta$ and $x$ are highly
dependant \emph{a posteriori}. In section \ref{sub:The-pseudo-marginal-approach}
we introduce PMCMC and describe how to apply it to MRFs, using the
exchange algorithm to avoid the intractable normalising constant.
The key component of a PMCMC algorithm is an efficient SMC sampler
for simulating from $p(x|\theta,y)$ and we describe candidate approaches
in section \ref{sub:SMC-samplers-for-1}.

\subsection{Intractable Normalising Constant\label{sub:Intractable-normalising-constant}}

Consider the use of a Metropolis-Hastings (MH) update on the full
conditional $p(\theta|x,y)\propto p(\theta)f(x,y|\theta)$. Let $\gamma(x,y|\theta)=\prod_{j=1}^{M}\Phi_{j}(x,y\in C_{j}|\theta)$ so that $f(x,y|\theta)=\frac{1}{Z(\theta)}\gamma(x,y|\theta)$. Using
a proposal $q(\theta^{*}|\theta)$, we obtain the acceptance probability:
\begin{equation}
1\wedge\frac{p(\theta^{*})\gamma(x,y|\theta^{*})q(\theta|\theta^{*})}{p(\theta)\gamma(x,y|\theta)q(\theta^{*}|\theta)}\frac{Z(\theta)}{Z(\theta^{*})},\label{eq:Z_accept}
\end{equation}
which cannot be calculated due to the presence of $Z(.)$, leading such models to be known as \emph{doubly intractable}
\citep{Murray2006}. A similar problem is encountered when using maximum
likelihood, or any other method that need to evaluate the likelihood.
There are several different classes of approaches for avoiding this:
pseudo-likelihood \citep{Besag1975}; variational approximations \citep{Murray2004};
Monte Carlo approximations \citep{Geyer1992b,Green2002,Atchade2008};
auxiliary variable approaches \citep{Moller2006,Murray2006}; and
ABC \citep{Grelaud2009a}\foreignlanguage{english}{.} Each of the methods
results in targeting an approximation to $p(\theta|x,y)$, unless
exact simulation from $f(x,y|\theta)$ is available. We focus on auxiliary
variable methods and ABC due to their ease of implementation and because
we find that using an MCMC algorithm in place of \foreignlanguage{english}{exact
simulation from $f(x|\theta)$ has little practical effect on the
examples we consider.}

\selectlanguage{english}%

\subsubsection{The Single Auxiliary Variable Method\label{sub:The-single-auxiliary}}

The first MCMC method that targets the true posterior in the presence
of the intractable normalising constant was the single
auxiliary variable method (SAVM), which originated in \citet{Moller2004} and \citet{Moller2006}.
Here we present an alternative justification of SAVM
to that in the original paper (given in \citet{Andrieu2010}), based on the principle that it is possible to design an MCMC algorithm
for simulating from some target $\pi(\text{\ensuremath{\theta})}$
when only unbiased positive estimates $\widehat{\pi}(\theta)$ of
an unnormalised version of $\pi$ are available. Specifically, a standard
Metropolis-Hastings algorithm that targets $\widehat{\pi}$, with
proposal $q(\theta^{*}|\theta)$ and acceptance probability
\begin{equation}
1\wedge\frac{\widehat{\pi}(\theta^{*})}{\widehat{\pi}(\theta)}\frac{q(\theta|\theta^{*})}{q(\theta^{*}|\theta)},\label{eq:approx_mh}
\end{equation}
has a stationary distribution of $\pi(\text{\ensuremath{\theta})}$.
Recall now the MH algorithm targeting $p(\theta|x,y)$ with
acceptance probability given in equation (\ref{eq:Z_accept}): equation (\ref{eq:approx_mh}) implies that a Metropolis-Hastings algorithm targeting an
unbiased estimate $\widehat{p}(\theta|x,y):=p(\theta)\gamma(x,y|\theta)\widehat{1/Z}(\theta)$
of $p(\theta|x,y)$, where $\widehat{1/Z}(\theta)$ is an unbiased
estimate of $1/Z(\theta)$, will have a stationary distribution of
$p(\theta|x,y)$.

SAVM targets a joint distribution $\pi(\theta,u|x,y)\propto q_{\mbox{u}}(u|\theta,x,y)p(\theta)\gamma(x,y|\theta)/Z(\theta)$, of which $p(\theta|x,y)$
is a marginal, constructed through the introduction of an auxiliary
variable $u$, where $ $$q_{\mbox{u}}(.|\theta,x,y)$ is an arbitrary distribution
and $u=(u_{x},u_{y})\in X \times Y$. This target is then simulated using
a Metropolis-Hastings algorithm on the $(\theta,u)$ space, with proposal
$\theta^{*}\sim q(.|\theta)$ and $u^{*}\sim f(.|\theta^{*})$ (the
simulation of $u^{*}$ being performed using exact sampling). The
acceptance probability for this move is given by:\foreignlanguage{english}{
\begin{equation}
1\wedge\frac{p(\theta^{*})\gamma(x,y|\theta^{*})q(\theta|\theta^{*})q_{u}(u^{*}|\theta^{*},x,y)\gamma(u|\theta)}{p(\theta)\gamma(x,y|\theta)q(\theta^{*}|\theta)q_{u}(u|\theta,x,y)\gamma(u^{*}|\theta^{*})}.\label{eq:sav}
\end{equation}
}

Comparing equations (\ref{eq:Z_accept}) and (\ref{eq:sav}) we see that
essentially SAVM is using single point importance sampling (IS)
estimates of $\widehat{1/Z}(\theta)$ and $\widehat{1/Z}(\theta^{*})$,
where $\widehat{1/Z}(\theta)=q_{u}(u|\theta,x,y)/\gamma(u|\theta)$
and $\widehat{1/Z}(\theta^{*})=q_{u}(u^{*}|\theta^{*},x,y)/\gamma(u^{*}|\theta^{*})$.
Linking this to the argument based on equation (\ref{eq:approx_mh})
gives us an alternative justification of the algorithm which suggests the following generalisation: \emph{any}
algorithm that finds an unbiased estimate of $\widehat{1/Z}(\theta)$
may be used in place of IS. For example, an SMC sampler
with $\pi_{1}=f(.|\theta)$ and $\pi_{T}=q_{u}(u|\theta,x,y)$ will
provide such an estimate whose efficiency is less dependent on the
choice of $q_{u}$ than is the IS estimate.  Note that SMC samplers play an important role later in the paper, and we refer the reader unfamiliar with the method to \citet{DelMoral2006} for further details.

\subsubsection{The Exchange Algorithm\label{sub:The-exchange-algorithm-2}}

In \citet{Murray2006}, the exchange algorithm is initially motivated
by the observation that in SAVM the ratio $Z(\theta)/Z(\theta^{*})$
is estimated indirectly, using the ratio of $\widehat{1/Z}(\theta^{*})$
and $\widehat{1/Z}(\theta)$. \citet{Murray2006} suggests that estimating
$Z(\theta)/Z(\theta^{*})$ directly may lead to an improved algorithm.
The resulting algorithm is actually most clearly formulated (for our
purposes) in an alternative fashion. Consider the target distribution $\pi(\theta,\theta^{*},u|x,y)=f(x,y|\theta)p(\theta)q(\theta^{*}|\theta)f(u|\theta^{*})/p(x,y)$ of which $\pi(\theta|x,y)$ is a marginal, where $u=(u_{x},u_{y})\in\mathcal{X}\times\mathcal{Y}$.
At iteration $i$, the exchange algorithm iterates the following two
operations (given input $(\theta(i-1),\theta^{*}(i-1),u(i-1))=(\theta,\theta^{*},u)$
from the previous iteration), giving an MCMC algorithm that targets
$\pi(\theta,\theta^{*},u|x,y)$:
\begin{enumerate}
\item Draw $\theta^{*}\sim q(.|\theta(i-1))$ and $u\sim f(.|\theta^{*})$.
\item Let $(\theta(i),\theta^{*}(i),u(i))=(\theta^{*},\theta(i-1),u)$ with
probability:
\begin{equation}
1\wedge\frac{p(\theta^{*})\gamma(x,y|\theta^{*})q(\theta|\theta^{*})\gamma(u|\theta)}{p(\theta)\gamma(x,y|\theta)q(\theta^{*}|\theta)\gamma(u|\theta^{*})}\frac{Z(\theta)Z(\theta^{*})}{Z(\theta^{*})Z(\theta)}=1\wedge\frac{p(\theta^{*})\gamma(x,y|\theta^{*})q(\theta|\theta^{*})\gamma(u|\theta)}{p(\theta)\gamma(x,y|\theta)q(\theta^{*}|\theta)\gamma(u|\theta^{*})},\label{eq:exchange_alg}
\end{equation}
otherwise set $(\theta(i),\theta^{*}(i),u(i))=(\theta(i-1),\theta^{*}(i-1),u)$.
\end{enumerate}
This deterministic move simply proposes a swap of $\theta$ and $\theta^{*}$,
hence the name {}``exchange algorithm''. The proof that this is
the appropriate acceptance probability for this move is given in \citet{Tierney1998},
where it is established that the acceptance probability for a deterministic
move $\theta \rightarrow T(\theta)$ on a target $\pi$ is $1\wedge\pi(T(\theta))/\pi(\theta)$. If we now compare equation (\ref{eq:Z_accept}) with equation (\ref{eq:exchange_alg})
we see that $\gamma(u|\theta)/\gamma(u|\theta^{*})$ can be thought
of as an IS estimate of $Z(\theta)/Z(\theta^{*})$,
so that the exchange method can be interpreted as using
an estimate of the acceptance probability. However, we note that in
this case we only have proof that this estimate in particular results
in an algorithm that targets the correct distribution - alternative
estimates that give a better estimate of this ratio may not be appropriate (although the similar approach of
\citet{Koskinen2008} is also shown to have the correct
target). The exchange method exhibits superior performance to SAVM in \citet{Murray2006}, and it is also easier
to implement - in SAVM there are more algorithmic choices
to make and these can severely affect the performance of the algorithm.
For these reasons we focus on the exchange algorithm for the remainder
of the paper, although we note that if the free choices in SAVM are well made (as in \citet{Moller2004}), it may outperform the exchange method.

Although the exchange algorithm targets the true posterior, it is
less efficient than a standard MH algorithm would be if the true target
were available. This inefficiency is accentuated if the proposed $u\sim f(.|\theta^{*})$
has only a small probability of being generated under $\theta$. \citet{Murray2006}
describes an extended exchange method to improve this inefficiency
whilst maintaining the exactness of the algorithm, using annealed IS \citep{Neal2001} as a substitute for the importance estimate
$\gamma(u|\theta)/\gamma(u|\theta^{*})$.  In the applications described in section \ref{sec:Applications} we found the use of this extended method to be essential in obtaining a reasonable performance for the methods described in the paper.  The central idea of the extended method is to move the proposed $u\sim f(.|\theta^{*})$  through a sequence of transitions so that it has a larger probability of being generated under $\theta$ (a similar idea may be used to improve SAVM in a similar way).  This is done by introducing a sequence of $K$ target distributions $f_k(.|\theta,\theta^*)\propto \gamma_k(.|\theta,\theta^*) = \gamma(.|\theta^{*})^{\beta_k}\gamma(.|\theta)^{1-\beta_k}$ where, for example, $\beta_k=(K-k+1)/(K+1)$ so that the sequence of targets provides a ``route'' from $f(.|\theta^{*})$ to $f(.|\theta)$. The extended exchange algorithm moves the initial point $u\sim f(.|\theta^{*})$ via a sequence of transitions $R_k(u'|u,\theta,\theta^*)$ that satisfy detailed balance with $f_k(u'|\theta,\theta^*)$, and changes the acceptance probability accordingly:

\begin{enumerate}
\item Draw $\theta^{*}\sim q(.|\theta(i-1))$ and $u_0\sim f(.|\theta^{*})$.
\item Apply the sequence of transitions: $u_1\sim R_1(u_1|u_0,\theta,\theta^*)$, $u_2\sim R_2(u_2|u_1,\theta,\theta^*)$, ... , $u_K\sim R_K(u_K|u_{K-1},\theta,\theta^*)$.
\item Let $(\theta(i),\theta^{*}(i),u(i))=(\theta^{*},\theta(i-1),u_K)$ with
probability:
\begin{equation}
1\wedge\frac{p(\theta^{*})\gamma(x,y|\theta^{*})q(\theta|\theta^{*})}{p(\theta)\gamma(x,y|\theta)q(\theta^{*}|\theta)}\prod_{k=0}^K \frac{\gamma_{k+1}(u_k|\theta,\theta^*)}{\gamma_{k}(u_k|\theta,\theta^*)} ,\label{eq:exchange_alg}
\end{equation}
otherwise set $(\theta(i),\theta^{*}(i),u(i))=(\theta(i-1),\theta^{*}(i-1),u_K)$.
\end{enumerate}

For our applications we have the factorisation $f(x,y|\theta)=f(x|\theta)g(y|x,\theta)$,
where $f(x|\theta)=\gamma(x|\theta)/Z(\theta)$ has an intractable
normalising constant and $g(y|x,\theta)$ is normalised. In this situation the exchange algorithm is
simplified slightly since $u_{y}$ does not need to be simulated.

\subsubsection{The Exchange Algorithm Without Exact Simulation\label{sub:The-exchange-algorithm}}

In step 1 of the exchange algorithm, on sampling $u$, exact simulation
from the likelihood is required. However, aside from a few special
cases (one of which is the Ising model, which can be sampled exactly using ``coupling from the past'' \citep{Propp1996}) this is not generally possible
for MRFs. \citet{Caimo2011} choose to approximate the exact simulation
by sampling $u$ from $f(.|\theta^{*})$ using an MCMC run that is
{}``long enough'' to get a point that can be treated as if it were
simulated exactly from $f(.|\theta^{*})$. We refer to this approach
as the \emph{approximate exchange algorithm}, and use $M$ to denote
the number of iterations of the MCMC algorithm used to simulate approximately from the likelihood. \citet{Caimo2011}
suggest that 500 iterations is a long enough run for models similar
to those studied in this paper, a conclusion supported by own study which suggests that as few as 50 or 100 iterations are usually
sufficient. In appendix B in the supplemental materials we give theoretical justification for the validity of this approach, proving that (using
a similar method to a proof in \citet{Andrieu2009}) when the MCMC
kernel for the exact exchange algorithm is uniformly ergodic, the
invariant distribution (when it exists) of the corresponding approximate
exchange algorithm becomes closer to the {}``true'' target (that
of the exact exchange algorithm) with increasing $M$. We also characterise
the rate of convergence of the approximate kernel. The same proof
justifies the use of an MCMC kernel as a substitute for simulating
exactly from the likelihood within SAVM.

\subsection{Marginal PMCMC\label{sub:The-pseudo-marginal-approach}}

\subsubsection{The Pseudo-Marginal Approach}

Our primary concern in this section is problem 3 in section \ref{sub:Posterior-inference}: that the data
augmentation approach to obtaining a sample from $p(\theta|y)$ is
inefficient when $\theta$ and $x$ are highly dependant \emph{a posteriori}.
The phase change in Ising models means that this dependence
is particularly clear in these models, but other MRFs (including
ERGMs \citep{Snijders2002}) also have this property. \selectlanguage{english}%
\citet{Beaumont2003} and \citet{Andrieu2009} describe an alternative
approach to sampling from $p(\theta|y)$ that is designed to avoid
the problems caused by this dependence: approximating
the {}``ideal'' algorithm that uses an MH update by replacing $p(\theta|y)$ with unbiased estimates of the form $\widetilde{p}(\theta|y)$ to $p(\theta|y)$. The argument
described in section \ref{sub:The-single-auxiliary} tells us that
such updates actually provide us with points from the desired target
$p(\theta|y)$. Such an approach is referred to as a \emph{pseudo-marginal}
approach. The efficiency of the MCMC chain based on these updates
depends on the variance of the estimator $\widetilde{p}(\theta|y)$.

The simplest useful estimator $\widetilde{p}(\theta|y)$ is an IS approximation $\widetilde{p}^{N}(\theta|y)=\frac{1}{N}\sum_{k=1}^{N}p(\theta,x^{(k)}|y)/q(x^{(k)}|\theta)$,
where $x^{(k)}\sim q(.|\theta)$. However, for the applications about
which we are interested, where $x$ is high-dimensional, it is difficult
to define a proposal distribution $q$ that results in an estimator
with a small variance. The optimal proposal is $p(x|\theta,y)$ (see \citep{Geyer2011}, for example): which we cannot sample from directly. The remainder of this section is devoted to using SMC samplers for this task. The framework of PMCMC in \citet{Andrieu2010} then tells us,
via the pseudo-marginal approach, how to use SMC samplers to more
efficiently obtain samples from $p(\theta|y)$. In section \ref{sub:Marginal-PMCMC}
we describe the marginal PMCMC framework, then in section \ref{sub:Exchange-marginal-PMCMC-1}
we describe the use of marginal PMCMC on MRFs.

\selectlanguage{british}%

\subsubsection{Marginal PMCMC\label{sub:Marginal-PMCMC}}

As in \citet{Andrieu2010}, in this section we make the assumption (relaxed in the following section) that the normalising constant $Z$ of $f(x,y|\theta)$
is independent of $\theta$ so that the joint posterior is given by:
$p(\theta,x|y)=p(\theta)\gamma(x,y|\theta)/Zp(y)$. We describe the marginal PMCMC algorithm briefly here - for a thorough
description see \citet{Andrieu2010}. The algorithm is an MCMC sampler
that targets $ $$p(\theta,x|y)$, operating on the factorisation
$p(\theta|y)p(x|\theta,y)$. The intuition behind the approach is
in devising a move on the joint $(\theta,x)$ space: in terms of moving
around the $\theta$ space it would be most efficient if the proposal
could take the form $q(\theta^{*}|\theta)p(x^{*}|\theta^{*},y)$,
so that $x^{*}$ is {}``perfectly adapted'' to the proposed $\theta^{*}$.
The approach is to devise an algorithm that approximates this idealised
situation, using an SMC sampler as a statistically efficient method
for simulating from an approximation $\widehat{p}(x^{*}|\theta^{*},y)$
to $p(x^{*}|\theta^{*},y)$. The point from $\widehat{p}(x^{*}|\theta^{*},y)$
is then used as a proposed point in a Metropolis-Hastings algorithm
that targets $ $$p(x^{*}|\theta^{*},y)$. \citet{Andrieu2010} derive
the acceptance probability that should be used by
explicitly writing down the target and proposal distributions involved. The
resulting algorithm in the case we describe here proceeds as follows.

At iteration $i$, beginning with the point $(\theta(i-1),x(i-1))$
and the estimate $\widehat{\phi}(i-1)$ outputted from iteration $i-1$:
\begin{enumerate}
\item Simulate $\theta^{*}\sim q(.|\theta(i-1))$, where $q$ is some proposal
distribution.
\item Run an SMC sampler on the $x$ space, with the final (unnormalised)
distribution as $\pi_{T}(x)=\gamma(x,y|\theta^{*})$. This gives a
particle approximation $\widehat{p}(x|\theta^{*},y)$ to $p(x|\theta^{*},y)$
and an estimate $\widehat{\phi}(\theta^{*},y)$ of its normalising
constant, $\phi(\theta^{*},y)=Zp(y|\theta^{*})$.
\item Sample a single point $x^{*}$ from $\widehat{p}(.|\theta^{*},y)$.
\item Set $(\theta(i),x(i),\widehat{\phi}(i))=(\theta^{*},x^{*},\widehat{\phi}(\theta^{*},y))$
with probability
\begin{equation}
1\wedge\frac{p(\theta^{*})\widehat{\phi}(\theta^{*},y)}{p(\theta)\widehat{\phi}(i-1)}\frac{q(\theta(i-1)|\theta^{*})}{q(\theta^{*}|\theta(i-1))},
\end{equation}
otherwise set $(\theta(i),x(i),\widehat{\phi}(i))=(\theta(i-1),x(i-1),\widehat{\phi}(i-1))$.
\end{enumerate}

Here note the interpretation of the algorithm as a Metropolis-Hastings
algorithm targeting an unbiased approximation to $p(\theta|y)\propto p(\theta)\phi(\theta,y)=p(\theta)Z\int_{x}f(x,y|\theta)dx$.

\subsubsection{Exchange Marginal PMCMC Algorithm\label{sub:Exchange-marginal-PMCMC-1}}

Now consider the application of marginal PMCMC to cases when the normalising
constant of $f(x|\theta)$ is a function of $\theta$,
and cannot be evaluated. In this case the joint posterior is $p(\theta,x|y)=p(\theta)\gamma(x,y|\theta)/Z(\theta)p(y)$
and direct application of the marginal PMCMC algorithm results in
the presence of the intractable term $Z(\theta)/Z(\theta^{*})$ in
the acceptance ratio. The combination of marginal PMCMC with the exchange
algorithm results in the disappearance of this ratio. Let us introduce
the target $\pi(\theta,\theta^{*},u|y)=p(\theta)p(y|\theta)q(\theta^{*}|\theta)\gamma(u|\theta^{*})/Z(\theta^{*})p(y)$,
where $u=(u_{x},u_{y})$, as in section \ref{sub:The-single-auxiliary}.
We then use the exchange algorithm as follows.
\begin{enumerate}
\item Draw $\theta^{*}\sim q(.|\theta(i-1))$ and $u\sim f(.|\theta^{*})$.
\item Run an SMC sampler on the $x$ space, with the final (unnormalised)
distribution as $\pi_{T}(x)=\gamma(x,y|\theta^{*})$ in order to obtain
the particle approximation $\widehat{p}(x|\theta^{*},y)$ and an estimate
$\widehat{\phi}(\theta^{*},y)$ of its normalising constant, $\phi(\theta^{*},y)=Z(\theta^{*})p(y|\theta^{*})$.
\item Sample a single point $x^{*}$ from $\widehat{p}(.|\theta^{*},y)$.
\item Let $(\theta(i),\theta^{*}(i),x(i),\widehat{\phi}(i),u(i))=(\theta^{*},\theta(i-1),x^{*},\widehat{\phi}(\theta^{*},y),u)$
with probability:
\begin{equation}
1\wedge\frac{p(\theta^{*})\widehat{\phi}(\theta^{*},y)q(\theta|\theta^{*})\gamma(u|\theta)}{p(\theta)\widehat{\phi}(i-1)q(\theta^{*}|\theta)\gamma(u|\theta^{*})}\frac{Z(\theta)Z(\theta^{*})}{Z(\theta^{*})Z(\theta)}=1\wedge\frac{p(\theta^{*})\widehat{\phi}(\theta^{*},y)q(\theta|\theta^{*})\gamma(u|\theta)}{p(\theta)\widehat{\phi}(i-1)q(\theta^{*}|\theta)\gamma(u|\theta^{*})},\label{eq:EMPMCMC}
\end{equation}
otherwise set $(\theta(i),\theta^{*}(i),x(i),\widehat{\phi}(i),u(i))=(\theta(i-1),\theta^{*}(i-1),x(i-1),\widehat{\phi}(i-1),u)$.
\end{enumerate}
The acceptance probability in equation (\ref{eq:EMPMCMC}) is again
derived using a proof similar to those in \citet{Andrieu2010}, in
conjunction with the result on deterministic transformations, from
\citet{Tierney1998}, used previously. This proof can be found in appendix A in the supplemental materials. An extension using the extended
exchange algorithm described in section \ref{sub:The-exchange-algorithm-2}
is trivial. The results in \citet{Andrieu2010} also tell us that
the $(\theta,x)$ points that are generated are from the joint posterior
$p(\theta,x|y)$, and that the unused SMC points generated from $\widehat{p}(.|\theta^{*},y)$
can be recycled in Monte Carlo estimates based on the $x$ space.

Our framework is now in place: the EMPMCMC algorithm addresses each
of the issues raised in section \ref{sub:Posterior-inference}. The
efficiency of the algorithm is heavily dependant on the design of
the SMC sampler on the $x$ space, and we examine this issue in the
following section.

\subsection{SMC Samplers for MRFs\label{sub:SMC-samplers-for-1}}

The use of SMC samplers for simulating from hidden Markov models is
well understood, but there have been few attempts to use them on more
general graphical models. The exception is the work of \citet{Hamze2005},
which they describe in application to discrete or Gaussian models
with a pairwise factorisation. It is their methods, \emph{hot coupling}
and \emph{tempering}, that form the basis of our approach. The most
fundamental choice in the design of such an algorithm is the choice
of targets $\pi_{1:N}$ to use, where the first target is easy to
simulate from, the final target is the desired distribution and the
targets between these two provide a {}``route'' from the first target
to the last. It is this choice, rather than matters such as designing
the forward and backward kernels that we focus on here.

Hot coupling proceeds by setting $\pi_{1}$ in the SMC sampler
to be a spanning tree of the true graph (which can be easily sampled,
and whose normalising constant may be exactly calculated using \citet{Carter1994})
and then to add edges to the graph, one at a time (randomly chosen),
until the true graph is reached at $\pi_{N}$ (this scheme could be
generalised to non-pairwise MRFs by forming larger cliques as the
SMC sampler progresses). Tempering consists of choosing the sequence
$\pi_{n}=f(x,y|\theta)^{1/t_{n}}$, where $(t_{n})_{n=1}^{N}$ is
a decreasing sequence of {}``temperatures'' with $t_{N}=1$. For
an Ising model this sequence of distributions has a simple interpretation,
since we obtain $f(x,y|\theta_{x},\theta_{y})^{1/t_{n}}\propto f(x,y|\theta_{x}/t_{n},\theta_{y}/t_{n})$.
\citet{Hamze2005} suggest that hot coupling is often more effective:
our own empirical results support this, as long as data generated from the initial tree
is a good approximation to data generated from the final target. This is difficult to quantify in advance in general: the case of the Ising model illustrates this, where the use of a square lattice exhibits the phase change described previously. Our empirical investigation was based on results obtained in our two applications in section \ref{sec:Applications}.  In the Ising model in the first application, we observed that data generated from an initial tree had similar characteristics to that that from the full grid (although this may not hold as strongly for larger grids).  In the social network application that follows, we did not find that there was a smooth transition between the initial tree and the final MRF.  A possible cause is that in this case the latent MRF has the circular structure exhibited in \citet{Frank1986}, and that the addition of edges that complete this circle change the character of the data drawn from the graph.  For this reason, the tempering method was preferred in this latter application.  The tempering method is likely to perform more reliably across a range of applications (it is also regularly used in non-MRF applications of SMC samplers), but for some MRFs hot coupling can be a useful tool.

\section{APPROXIMATE BAYESIAN COMPUTATION\label{sec:Approximate-Bayesian-computation}}

\subsection{Introduction}

Approximate Bayesian computation (ABC) is a method designed, by \citet{Pritchard1999},
to perform Bayesian inference in cases where a likelihood $l(y|\theta)$
cannot be evaluated because it is unavailable or computationally intractable.
The basic idea is to approximate the likelihood by $\hat{l}_{\epsilon}(y|\theta)=\int_{y'}l(y'|\theta)\pi_{\epsilon}(y'|y)\mathrm{d}y$ where $\pi_{\epsilon}(y'|y)$ is a probability density on the same
space as $y$, and centred on $y$ with a concentration determined
by $\epsilon$. If $\epsilon=0$, then $\hat{l}_{\epsilon}(y|\theta)$
is equal to the true likelihood $l(y|\theta)$. This approximate likelihood
can be evaluated by a Monte Carlo approximation $\widetilde{l}_{\epsilon}(y|\theta)=\frac{1}{R}\sum_{r=1}^{R}\pi_{\epsilon}(y'^{(r)}|y)$ where $y'^{(r)}\sim l(.|\theta)$ and $R>0$. This likelihood
gives more weight to values of $\theta$ that are more
likely to generate data similar to $y$.

In practice, due to the possible high dimension of the data, a summary
statistic $S(y)$ is usually used in place of the data, giving a further
approximation (if the statistic is not sufficient) to the true likelihood
using:
\begin{equation}
\hat{l}_{\epsilon,S}(y|\theta):=\hat{l}_{\epsilon}(S(y)|\theta):=\int_{y'}l(S(y')|\theta)\pi_{\epsilon}(S(y')|S(y))\mathrm{d}y'.\label{eq:abc_approx_S}
\end{equation}

MCMC \citep{Marjoram2003a} and SMC samplers \citep{Sisson2007smc,DelMoral2008,Robert2011smc}
have been introduced to enable efficient exploration of the $\theta$
space when using an ABC approximation to the likelihood. In section
\ref{sec:Applications} we apply an ABC-SMC sampler to the problem
of parameter estimation in MRFs. The main advantage of this approach
is that it avoids the difficulties listed in section \ref{sub:Posterior-inference}.

For some of the models we consider we can take $\epsilon=0$, and
low-dimensional sufficient statistics sometimes exist: in these cases
it can be possible to devise an ABC algorithm that targets the correct
posterior distribution $p(\theta|y)$. However this is not the case in general. For higher dimensional
sufficient statistics it can require more computational effort to
obtain a useful sample when $\epsilon=0$. In some cases using $\epsilon=0$
is impracticable, and sufficient statistics are not available, and
in these cases it is only possible to obtain a sample from an approximation
to the true posterior with ABC. Usually the use of ABC methods involves
a trade off between the degree of approximation to the true posterior
and the computational effort required to obtain the sample. Thus tuning
ABC can be difficult and it is not easy to quantify the effect of
the approximation to the true posterior that is used. The intricacies
of using ABC are discussed further in the review of \citet{Marin2011}.

\subsection{Application to MRFs\label{sub:Application-to-MRFs}}

ABC has previously been applied to inferring the parameters of MRFs \citep{Grelaud2009a} - here we instead consider noisy
or incomplete data. 
It is particularly simple to apply ABC to the models that we focus
on in this paper, since both Ising models and ERGMs are defined in
terms of statistics of the data. However, as is the case when using the exchange algorithm on these problems,
it is not possible to exactly simulate from $l(.|\theta)$ and we
consider the effect of using for example, MCMC (as in \citet{Grelaud2009a})
as a substitute. The proof in appendix B in the supplemental materials
describes the effect of using an MCMC run of length $M$ for simulating
from $l(.|\theta)$ within an ABC-MCMC algorithm (as in \citet{Marjoram2003a}).
The result is analogous to that obtained for the exchange algorithm
and SAVM: the invariant distribution (when it exists) of
this method becomes closer to the {}``true'' target (the invariant
distribution of the standard ABC-MCMC algorithm) with increasing $M$.

In our models $f(x,y|\theta)$ factorises as $f(x,y|\theta)=f(x|\theta)g(y|x,\theta)$.
For a given $\theta$ the simulation $y'$ from the likelihood $l(.|\theta)$
is performed through first simulating $x'$ from $f(.|\theta)$, then
by simulating $y'$ from $g(.|x',\theta)$. In this situation we note
that $x'$ is always proposed from its prior $f(x|\theta)$, as opposed
to the posterior $p(x|\theta,y)$, therefore when the effect of the
data dominates that of the prior the exploration of the $x$ space
is not likely to be as efficient as that used as in the PMCMC approach
described in the previous section.

\section{APPLICATIONS\label{sec:Applications}}

\subsection{Methods}

In this section we apply each of the methods described in the paper
to two different data sets. The configuration of the algorithms that
we use has some commonality between each case, thus we begin by describing
these common aspects of the algorithms before describing the specifics
of their application in the relevant sections. Our implementation
is in Matlab, and makes use of the UGM package (available from Mark Schmidt's
website at \texttt{http://www.di.ens.fr/\textasciitilde{}mschmidt/Software/UGM.html}).

\subsubsection{Approximate Bayesian Computation}

\selectlanguage{english}%
The choices made in constructing the approximate likelihood in our ABC algorithms are always the same, up to the choice of summary statistics: we use the uniform kernel $\pi_{\epsilon}(S(y')|S(y))\propto\mathbf{1}_{\left\Vert S(y')-S(y)\right\Vert <\epsilon}$
as the data comparison function, where $\left\Vert .\right\Vert $
is the Euclidean norm.

Our choice of the uniform kernel allows us to use the adaptive ABC-SMC sampler of\foreignlanguage{british}{ \citet{DelMoral2008}
to generate weighted points from the posterior. This method adaptively
chooses a sequence of targets in the SMC sampler, where the $n$'th
target is the ABC posterior using the likelihood in equation (\ref{eq:abc_approx_S})
with tolerance $\epsilon_{n}$. The sequence $(\epsilon_{n})$ is
chosen by, for each $n$, taking the smallest tolerance that ensures
that a particular percentage of the particles is given non-zero weight. We always use 10000 particles, initialise the
ABC-SMC with $\epsilon_{1}=20$ and terminate it when $\epsilon_{n}$
is reduced to zero,} and resample at every iteration (which
is likely to be the most efficient option when using this ABC algorithm, since otherwise particles with zero weight are carried from one iteration of the sampler to the next)
using stratified resampling. This method uses an ABC-MCMC move as the forward
kernel within the SMC sampler, and we choose a MH move with a random
walk proposal.

\subsubsection{DA and PMCMC}

\selectlanguage{british}%
In our implementation of DA, the update of $\theta$ uses an MH step
whose proposal $q$ is a random walk with variance $sI$, where $s$ is some problem
specific scaling and $I$ is the identity matrix of the appropriate
dimension. To update $x$ we use $L$ sweeps of the single site Gibbs
sampler.

Our configuration of the marginal PMCMC algorithm is chosen to facilitate
easy comparison with the DA approach above. The forward kernel in
the SMC sampler is always chosen to be the single site Gibbs sampler,
and stratified resampling was performed when the effective sample size (ESS) dropped below
0.5 multiplied by the number of particles. To ensure that the computational
effort of a PMCMC iteration is comparable to a sweep of the DA algorithm,
we take the number of particles $P=\left\lfloor L/T\right\rfloor $,
where $T$ is the number of targets used in the SMC sampler within
the PMCMC. In every case, our figures show 4500 points generated by the algorithm,
after a burn in of 500 iterations.

\selectlanguage{english}%

\subsubsection{Simulation From the Likelihood and Exchange Algorithm}

The \foreignlanguage{british}{extended exchange algorithm (section
\ref{sub:The-exchange-algorithm}) is used when updating $\theta$
in both the DA and PMCMC algorithms; $B$ bridging targets are} used.
Both this use of the exchange algorithm and the ABC approach require
simulation from the likelihood $f(x|\theta)$. In all cases we use
1000 sweeps of the single site Gibbs sampler described in section
\ref{sub:Gibbs-samplers-for}. Note that the proof in appendix B in the supplemental materials
indicates that there is no restriction on choice of the initial point
$x_{0}$ of this Gibbs sampler. We simply chose $x_{0}=y$ (which
we are free to do since in this case $x$ and $y$ inhabit the same
space). This choice, where the same initial value is used at each
iteration, could potentially be improved by choosing the initial value
based on the previous run of the Gibbs sampler, however this has little
practical effect in our applications. Changing the initial value for each run of the Gibbs sampler necessitates
an alteration to the proof in the appendix, and this alteration is described in the appendix.

\selectlanguage{british}%
We note that all the methods contain the single site Gibbs sampler
(either for simulating from the likelihood, or for updating $x$),
which we know to be inefficient. More efficient MCMC moves could be
used in its place, for example the Swendsen-Wang algorithm for Ising
models \citep{Higdon1998} or, for the ERGM example in the next section,
the {}``tie no tie'' sampler used in \citet{Caimo2011}. This would
improve the efficiency of all of these algorithms, and for real applications
we would advocate such an approach. Here we have chosen not to investigate
these potential efficiency gains, concentrating instead on the impact
of avoiding the use of the DA algorithm.

\selectlanguage{english}%

\subsubsection{Reporting of Results}

The difference in performance of the algorithms in both applications
is large and can be observed through visualising the samples that are generated. Estimates of posterior means and variances
do not provide any useful information above plots of the samples,
and measures of the efficiency of the MCMC such as autocorrelation
times are not informative since some of the chains we wish to compare
are not close to stationarity. Thus we chose to only
represent the output of the samplers through plots of the samples
they produce (in the ABC-SMC algorithm the positions of the equally weighted particles after resampling are shown).

\selectlanguage{british}%

\subsection{Ising Model\label{sub:Markov-random-field}}

In this section we apply the methods described in the paper to inference
of the parameters of an Ising model. We use noisy observations of
a hidden $10\times10$ pixel two-dimensional grid, simulated from
the model in equation (\ref{eq:general_gm}) with parameters $\theta_{x}=0.1$
and $\theta_{y}=0.1$. This data is shown in figure \ref{fig:Observed-data.}.
Our prior over both $\theta_{x}$ and $\theta_{y}$ is uniform on
the interval $[0,3]$, and our model for the data is given by equation
(\ref{eq:general_gm}). Our goal is to simulate from the posterior $p(\theta_{x},\theta_{y}|y)$,
and in this section we compare ABC, DA and PMCMC algorithms.

Our ABC algorithm used two summary statistics of the data: \foreignlanguage{english}{$S_{1}(y)=\sum_{(i,j)\in\mathbf{N}}y_{i}y_{j}$
(the number of equivalently valued }neighbours\foreignlanguage{english}{)
and $S_{2}(y)=\sum_{i}y_{i}$ (the magnetisation)}. In the adaptive
ABC-SMC algorithm, we choose that 70\% of the particles were given
non-zero weight at each iteration. Points from $p(\theta_{x},\theta_{y}|y)$
using this method, which was terminated at  $n=12$, are shown in figure \ref{fig:Weighted-points-from}: the posterior
mass is distributed around the areas where $\theta_{x}$ is small
and/or $\theta_{y}$ is small. It is clear that this posterior will
pose a significant challenge to DA: in order to explore the posterior
fully it is necessary to sample both large and small values of $\theta_{x}$,
and moving from one side of the critical value of $\theta_{x}$ to
the other is difficult due to the dependence between $x$ and $\theta_{x}$. Points from $p(\theta_{x},\theta_{y}|y)$ using DA, in which we chose $s=1$ and $L=10^{5}$, are shown in figure \ref{fig:Points-simulated-from}.
The posterior dependency between $\theta_{x}$ and $x$, and the inefficiency
of the updates on $p(x|\theta,y)$ (even when using $10^{5}$ sweeps
of the Gibbs sampler at every iteration), inhibit the sampler from
moving below the critical value of $\theta_{x}$.

In this implementation of the marginal PMCMC algorithm, we use hot
coupling for the SMC sampler updates, adding a single edge at each
target (this results in a total of 82 targets, giving approximately
1200 particles). The simlated points are shown in
figure \ref{fig:Points-simulated-from-1}: the sampler
has explored the whole posterior, with no evidence of difficulty in
passing the critical value of $\theta_{x}$. Figure \ref{fig:Trace-plot-of}
shows the trace plot of $S_{1}(x)$ for both the DA and the PMCMC
methods, illustrating that the PMCMC approach enables exploration
of the $x$ space, whereas DA only explores graphs that correspond
to a large value of $\theta_{x}$. In both the DA and PMCMC extended
exchange algorithms, we use $B=100$ intermediate targets in the annealed
IS: fewer intermediate targets can result in poor
estimation of the ratio of normalising constants, and therefore inefficiency
in both algorithms.  If the original exchange algorithm is used ($B=1$), the results from the PMCMC method do not look dissimilar to those produced by the DA method: when there is a proposed change in $\theta_x$ that crosses the critical value, the move is rejected since the proposed $u\sim f(.|\theta^*)$ has a small probability of being generated under $\theta$.

The dominating factor in the computational cost of each algorithm is in the Gibbs sampler for the simulation from the $x$ space (with the target of either $p(x|\theta)$ or $p(x|\theta,y)$). For each iteration of the DA or PMCMC algorithm $10^5$ sweeps of the Gibbs sampler are performed at each iteration of the MCMC (along with an additional $10^3$ for the exchange algorithm), so the total cost of our run is approximately $5\times 10^8$ sweeps.  For each target of the ABC-SMC $10^7$ sweeps are performed, so the total computational cost is $1.2\times 10^8$ sweeps.

\begin{figure}[h]
\caption{Results for Bayesian parameter estimation in a hidden Ising model.\label{fig:Results-for-Bayesian}}

\subfloat[Observed data.\label{fig:Observed-data.}]{

\includegraphics[scale=0.35]{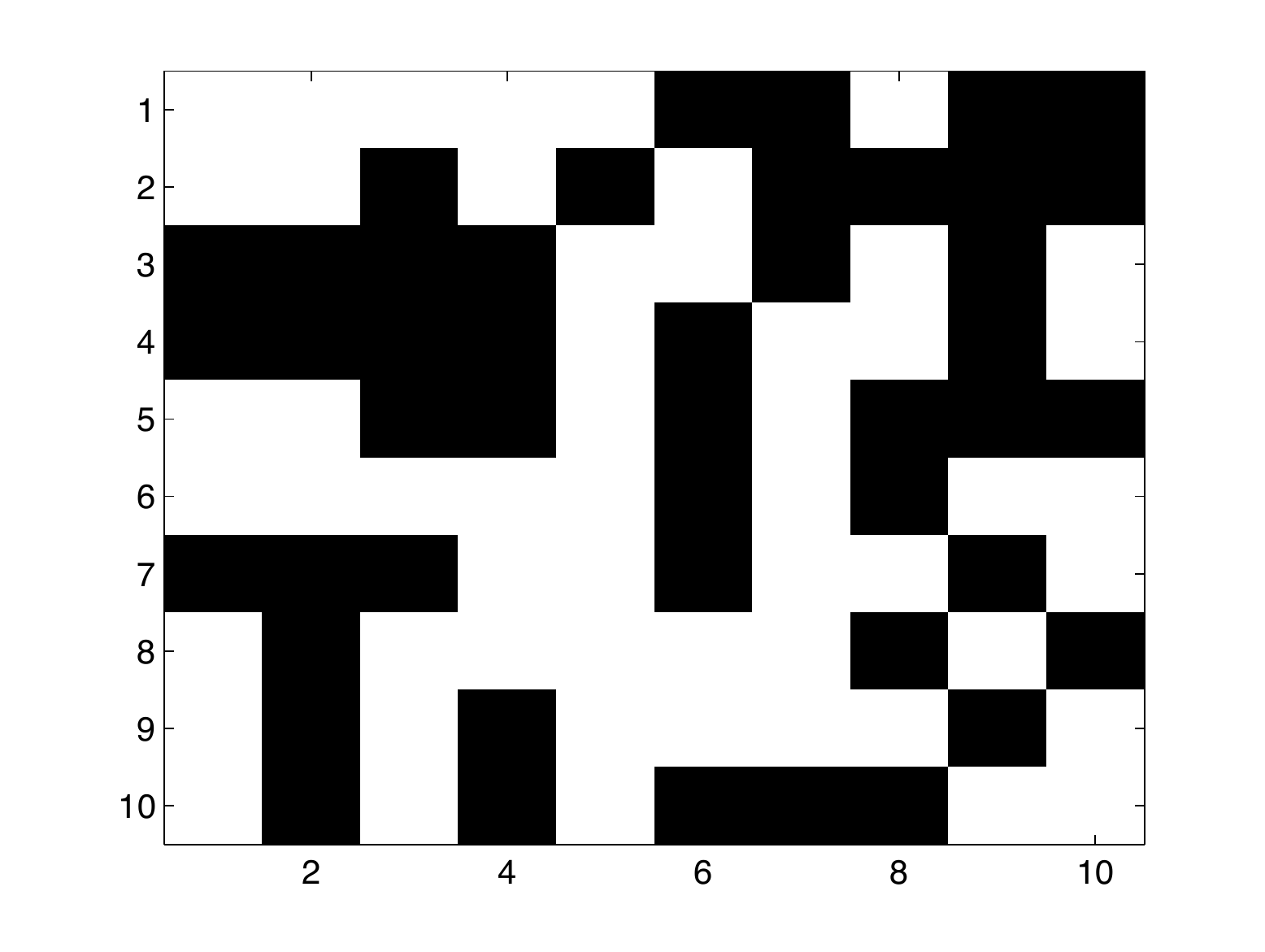}}\subfloat[Points from the posterior produced by the ABC-SMC algorithm.\label{fig:Weighted-points-from}]{

\includegraphics[scale=0.35]{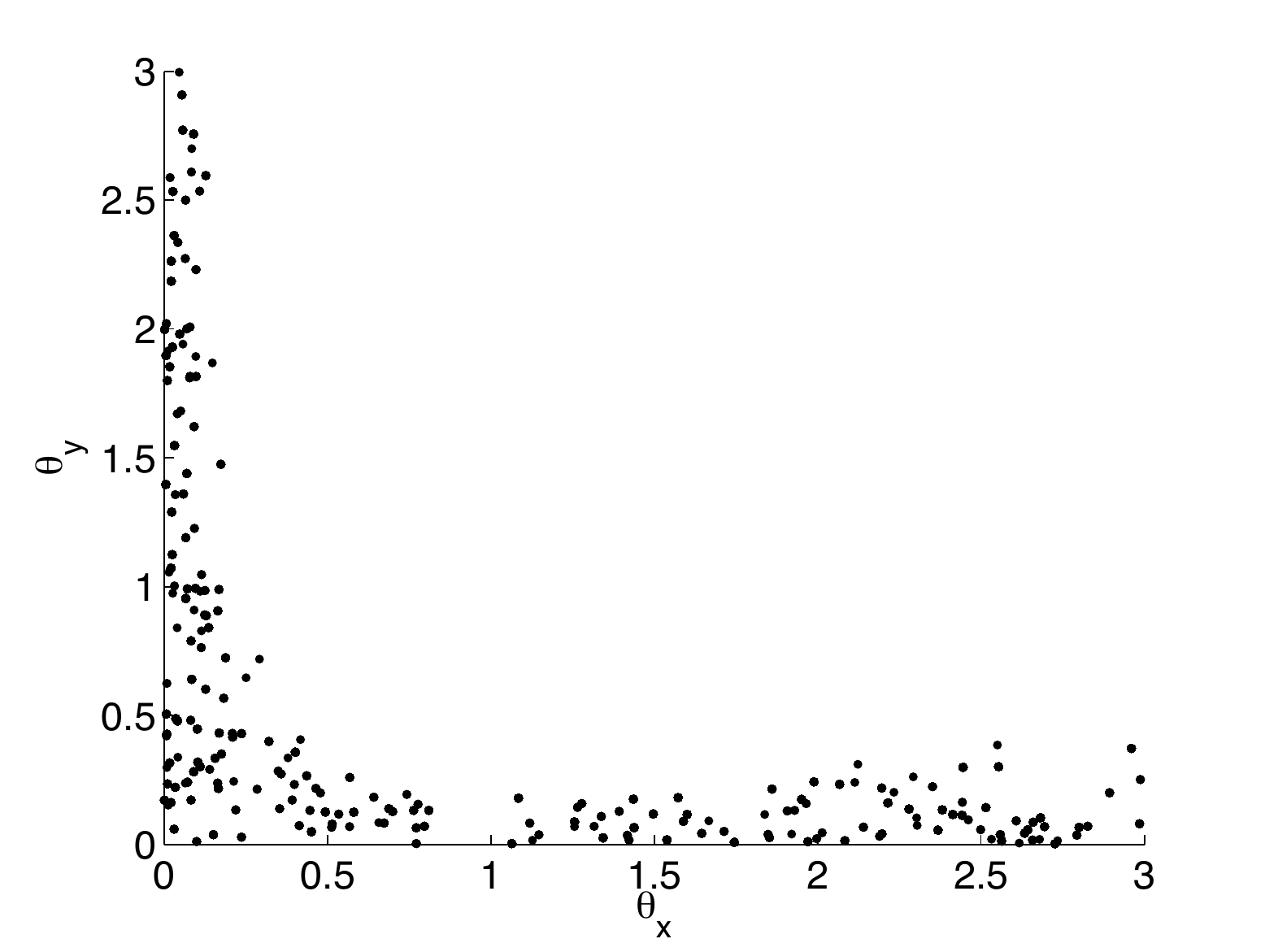}}\subfloat[Points simulated from the posterior using the DA exchange algorithm.\label{fig:Points-simulated-from}]{

\includegraphics[scale=0.35]{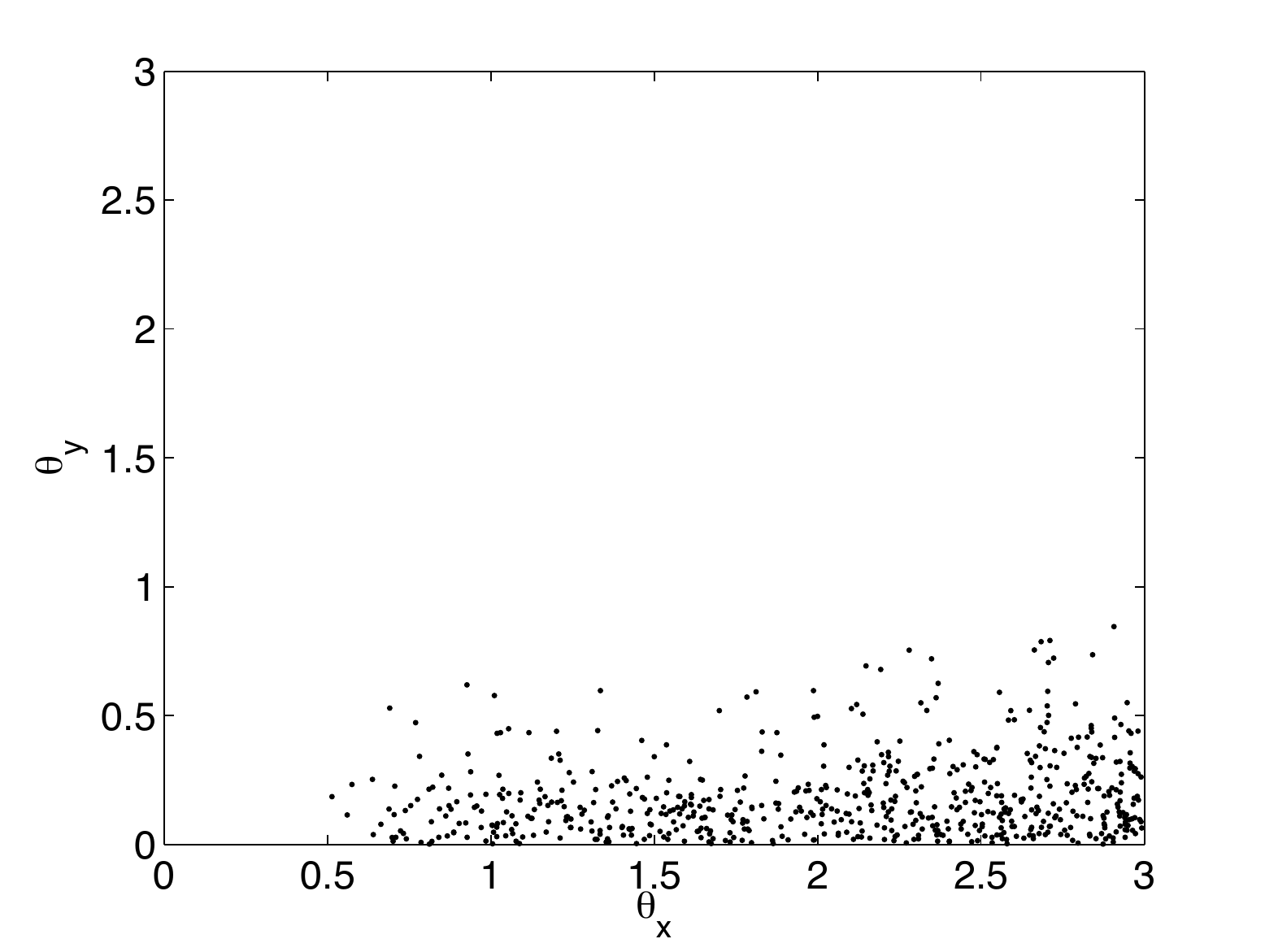}}

\subfloat[Points simulated from the posterior using the exchange marginal PMCMC
algorithm.\label{fig:Points-simulated-from-1}]{

\includegraphics[scale=0.35]{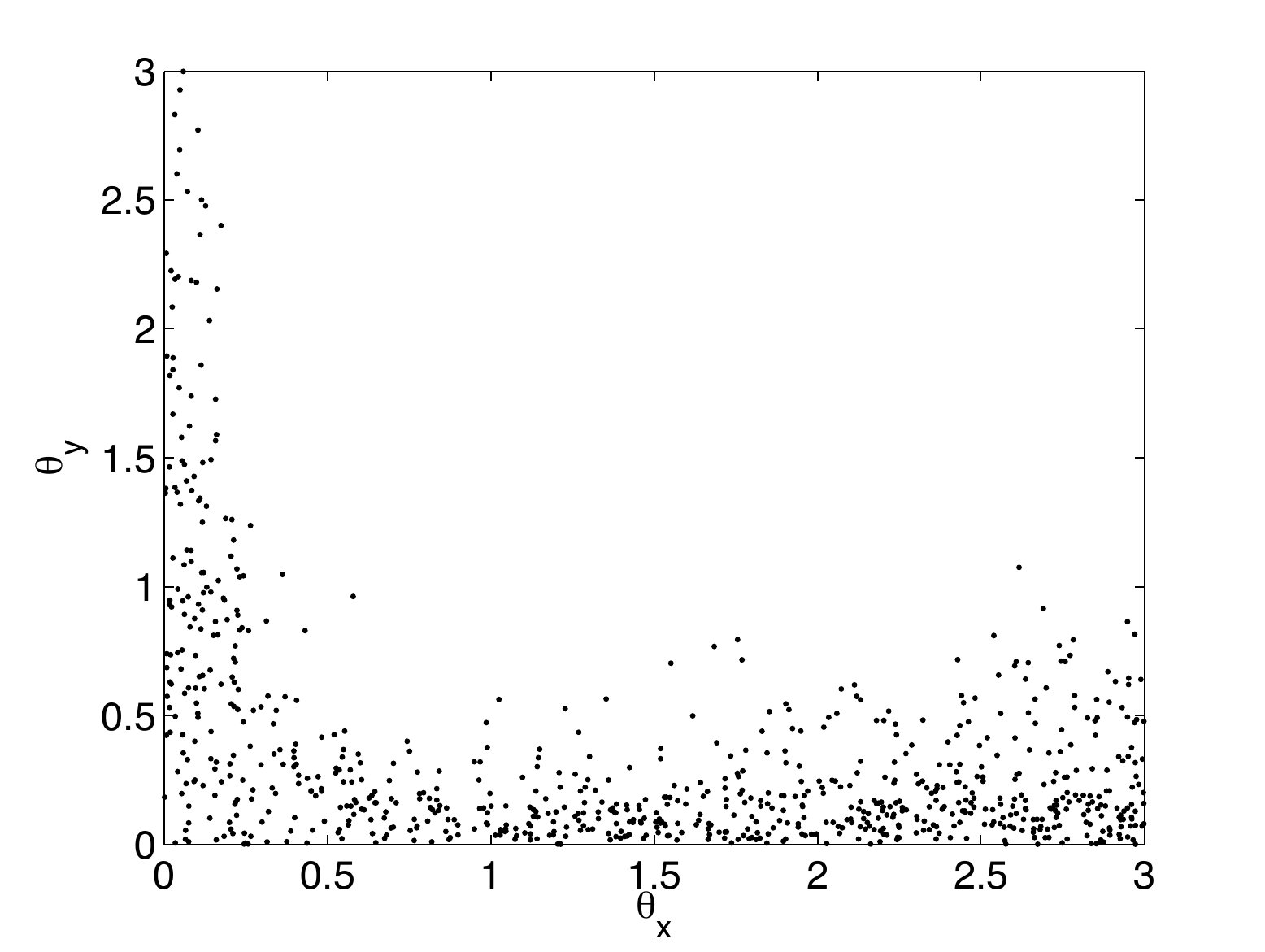}}\subfloat[Trace plot of the statistic $S_{1}(x)$, the number of identical neighbours
in the hidden field $x$, for the DA (black) and the PMCMC (grey)
algorithms.\label{fig:Trace-plot-of}]{

\includegraphics[scale=0.44]{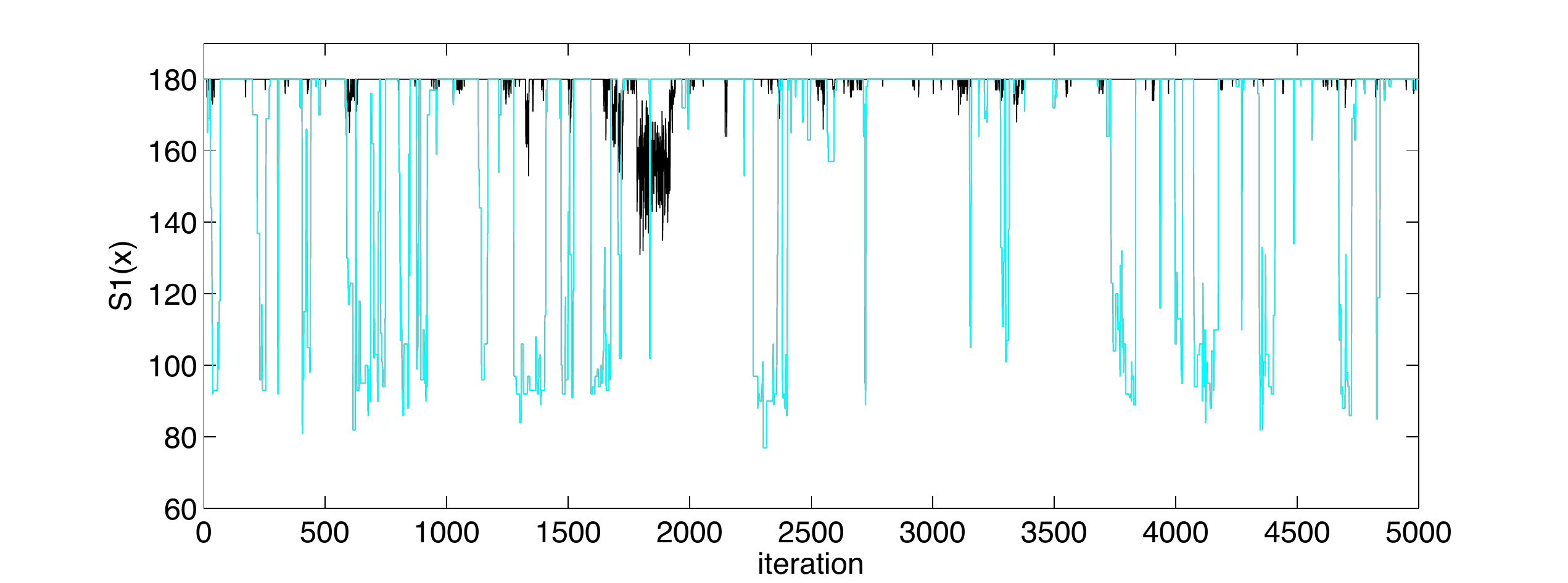}}
\end{figure}

\subsection{Social Network Data\label{sub:Social-network}}

We consider the application of our methods to the Florentine family
business graph studied in \citet{Caimo2011}, shown in figure \ref{fig:Observed-data.-1}.
We begin by assuming the network $x$ is directly observed and use
exactly the model for the data as that used in \citet{Caimo2011}:
using an ERGM (equation (\ref{eq:ergm})) with $S_{1}(x)=\sum_{i<j}x_{ij}$
(the number of edges) and $S_{2}(x)=\sum_{i<j<k}x_{ik}x_{jk}$ (the
number of 2-stars) , and prior on $\theta_{x}=(\theta_{1},\theta_{2})$
as $\theta_{x}\sim\mathcal{N}(0,30I_{2})$.

We begin by examining ABC as a direct alternative to the MCMC approach
in \citet{Caimo2011} (the DA and PMCMC approaches are not applicable here since there is no latent space). We use the statistics $S_{1}$ and $S_{2}$ as our summary of the data, and since these statistics
are sufficient, when $\epsilon=0$ the ABC posterior is equivalent
to the true posterior. In this implementation the scale of the target
changes dramatically across the iterations (since the posterior is
significantly tighter than the prior), thus we adaptively choose the
proposal variance in the MH move\foreignlanguage{english}{, so that
at iteration $n$, the variance is $2\widehat{\Sigma}_{n-1}$ where
$\widehat{\Sigma}_{n-1}$ is the sample variance of the particles
at iteration $n-1$ (as in \citet{Robert2011smc}). We choose that}
50\% of the particles were given non-zero weight at each iteration.
The ABC-SMC terminated at $n=16$, and weighted points drawn from
$p(\theta_{1},\theta_{2}|y)$ using this method are shown in figure
\ref{fig:Weighted-points-from-1}. These points are in
good agreement with the shape of the posterior shown in \citet{Caimo2011}.
We note that the {}``population'' nature of the SMC algorithm acts
as a substitute for the population MCMC method employed in
that paper.

Now consider the case where it is assumed that the observed graph,
now denoted by $y$, is a noisy observation of some underlying
graph $x$. Specifically, we use the model in equation (\ref{eq:ergm_noise}),
which account for noisy observations of the edges of the underlying
graph. We use a generalisation of the model described above,
defined on the extended parameter $\theta=(\theta_{1},\theta_{2},\theta_{y})$
with the same priors on $\theta_{1}$ and $\theta_{2}$, and with
a prior of $\theta_{y}\sim\mathcal{U}[0,100]$ on the additional parameter
$\theta_{y}$. Note that the data itself does not suggest that this
model is particularly suitable: it is chosen simply to highlight the
computational problems that can result in the presence of a latent
ERGM. We use the same ABC-SMC algorithm as above, using the same summary
statistics (which are now not sufficient). The ABC-SMC again terminated
at $n=16$, and weighted points simulated from $p(\theta_{1},\theta_{2}|y)$
using this method are shown in figure \ref{fig:Weighted-points-from-2}.
One of the limitations of ABC is evident here: since the statistics
are not sufficient, it is difficult to assess how accurate the approximation
to the true posterior is, even though we have $\epsilon=0$. We applied the DA algorithm and PMCMC algorithms to the same model
that assumes the data is noisy (including the parameter $\theta_{y}$).
In the DA algorithm we chose $s=10$ and $L=10^{4}$ and points generated
from $p(\theta_{1},\theta_{2}|y)$ are shown in figure \ref{fig:Points-simulated-from-2}.
The posterior sample produced by this sampler was highly dependant
on the initial point. In this particular example, when the initial
value of $x$ is the complete graph (where all edges are present),
the chain gets stuck in this state, which is highly correlated \emph{a
posteriori} with a small value of $\theta_{y}$. Again, the inefficiency
of the updates on $x$ and the dependency between $x$ and $\theta$
is seen to lead to the poor performance of the DA approach.

In addition to comparing the PMCMC approach to DA, we also compare
the original IS based pseudo-marginal approach to
the more general PMCMC approach in order to observe the benefit of
using an SMC sampler in sampling the $x$ space. Specifically, we
compare IS with $10^{4}$ importance points drawn
using a uniform distribution independently on each edge, with the
SMC sampler with tempering using 1000 targets and 10 particles. Figures
\ref{fig:Points-simulated-from-3} and \ref{fig:Points-simulated-from-4}
respectively illustrate the PMCMC results using these two schemes.
The performance of the IS based technique is poor
since the importance proposal is poor and does not provide an accurate
estimate of the normalising constant $Z(\theta)p(y|\theta)$. The
performance of the tempering SMC based EMPMCMC is significantly better
than both the IS and DA approaches, with convergence
to the region shown in figure \ref{fig:Points-simulated-from-4} regardless
how the MCMC is initialised. This lack of dependence on initial conditions,
and the free exploration of the $x$ and $\theta_{y}$ space that
is allowed by the method, provide some evidence that this method has
found the true posterior.

In the DA and PMCMC extended exchange algorithms, we use 1000 intermediate
targets in the extended exchange algorithm: again we find that fewer
intermediate targets results in poor estimation of the ratio of
normalising constants and thus inefficient MCMC algorithms.

Again, the dominating factor in the computational cost of each algorithm is in the Gibbs sampler for the simulation from the $x$ space. In this case, for each iteration of the DA or PMCMC algorithm $10^4$ sweeps of the Gibbs sampler are performed at each iteration of the MCMC (along with an additional $10^3$ for the exchange algorithm), so the total cost of our run is approximately $5\times10^7$ sweeps.  For each target of the ABC-SMC $10^7$ sweeps are performed, so the total computational cost is $1.6\times 10^8$ sweeps.  The method of \citet{Caimo2011} is relatively cheap compared to ABC-SMC, with the only simulation from the $x$ space being carried out as part of the exchange algorithm (which costs $10^3$ sweeps).

\begin{figure}[h]
\caption{Results for Bayesian parameter estimation in a hidden ERGM.\label{fig:Results-for-Bayesian-1}}

\subfloat[Observed data.\label{fig:Observed-data.-1}]{

\includegraphics[scale=0.26]{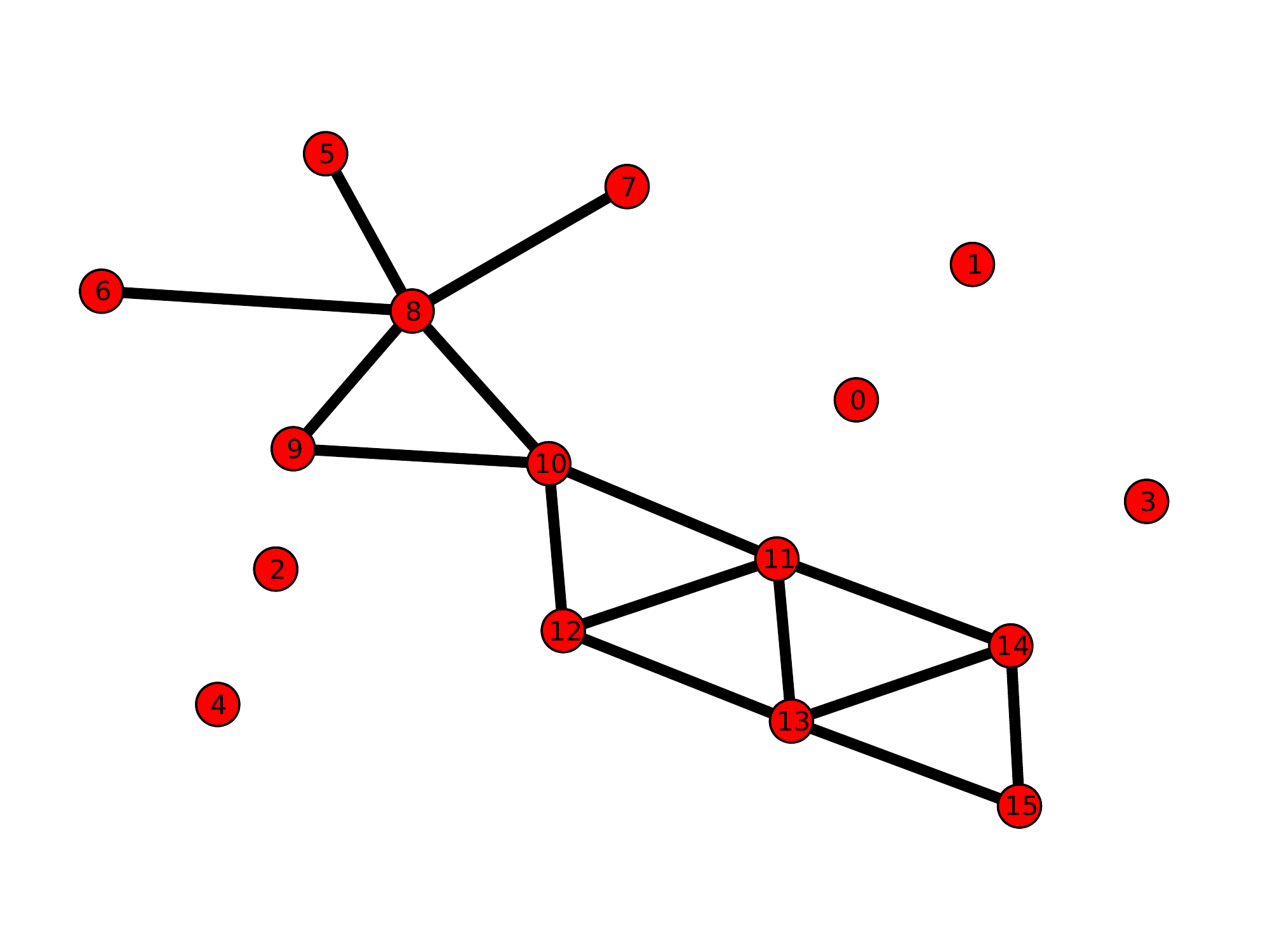}}\subfloat[Points from the posterior (with no observation error) produced
by the ABC-SMC algorithm.\label{fig:Weighted-points-from-1}]{

\includegraphics[scale=0.35]{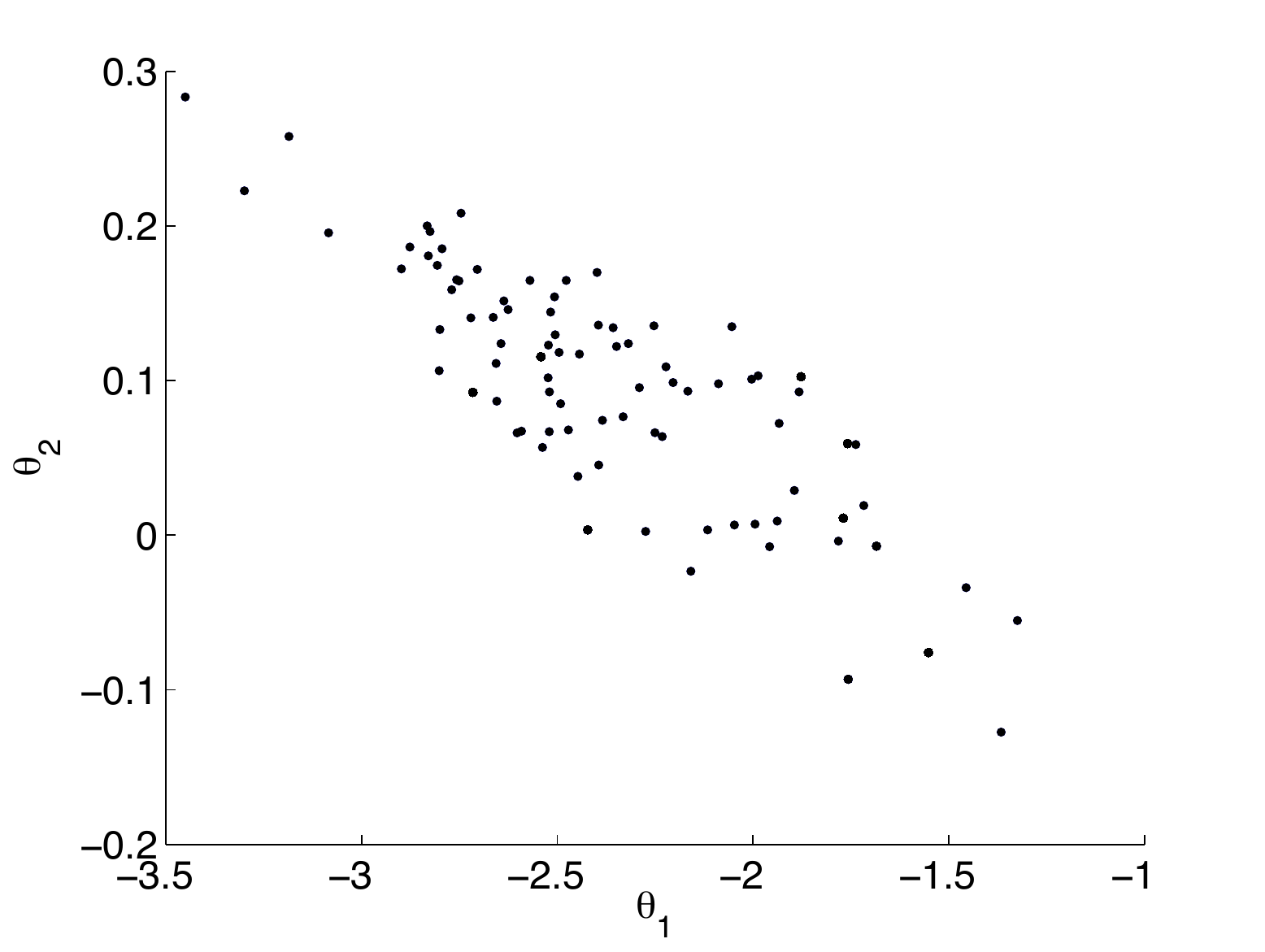}}\subfloat[Points from the posterior produced by the ABC-SMC algorithm.\label{fig:Weighted-points-from-2}]{

\includegraphics[scale=0.35]{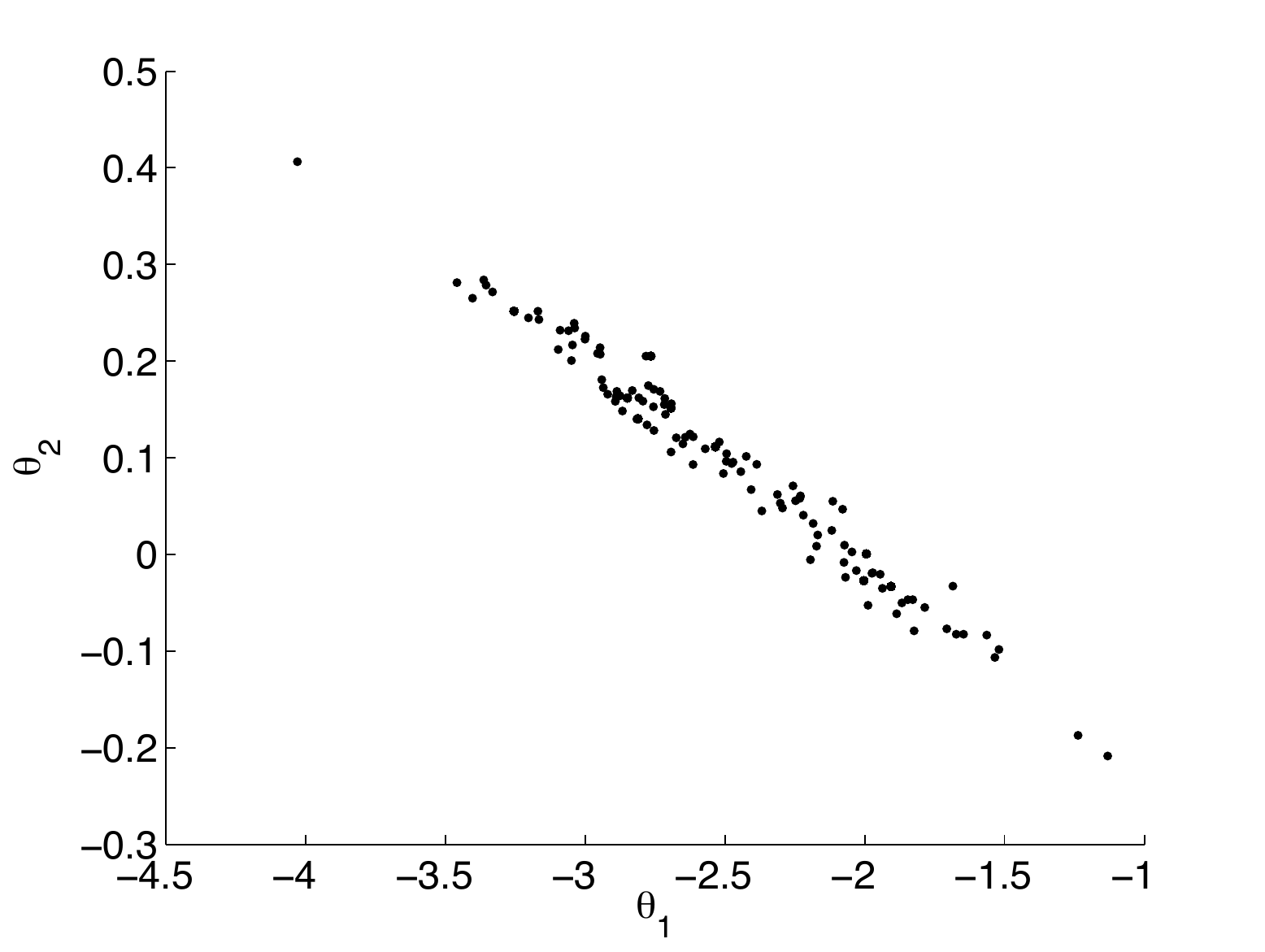}}

\subfloat[Points simulated from the posterior using the DA exchange algorithm.\label{fig:Points-simulated-from-2}]{

\includegraphics[scale=0.35]{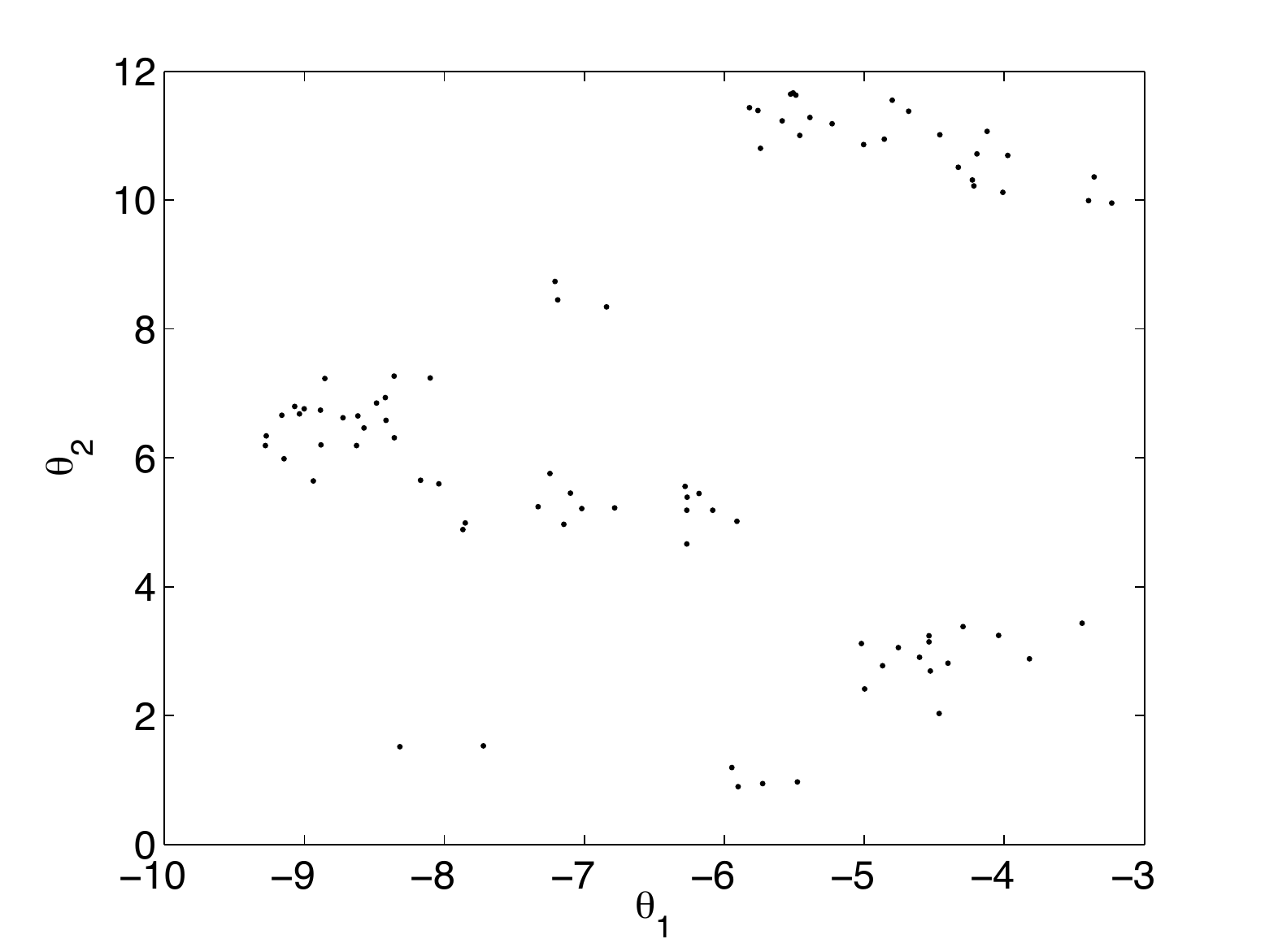}}\subfloat[Importance sampling based pseudo-marginal approach, with the exchange algorithm.\label{fig:Points-simulated-from-3}]{

\includegraphics[scale=0.35]{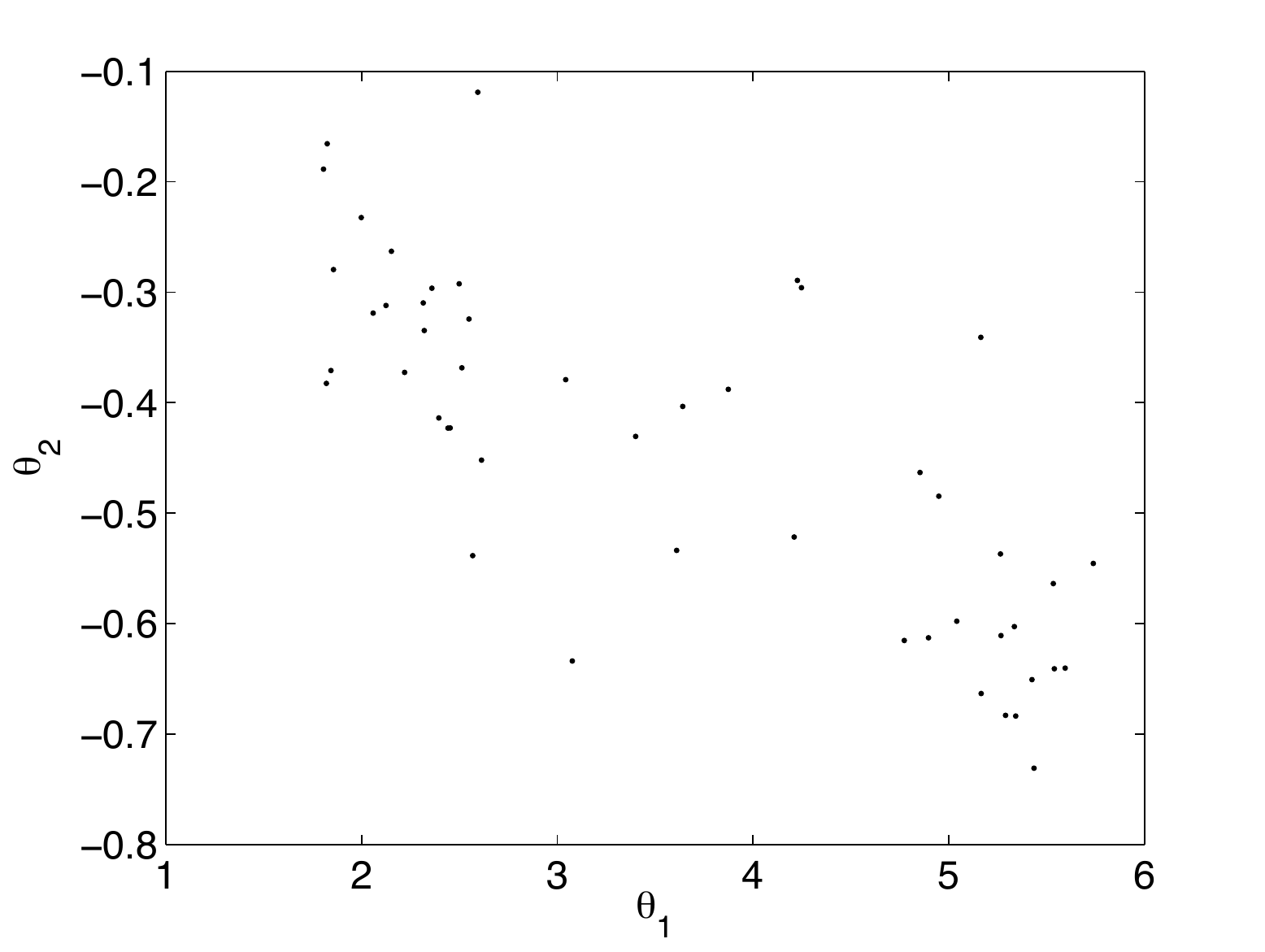}}\subfloat[Points simulated from the posterior using the exchange marginal PMCMC
algorithm containing an SMC sampler.\label{fig:Points-simulated-from-4}]{

\includegraphics[scale=0.35]{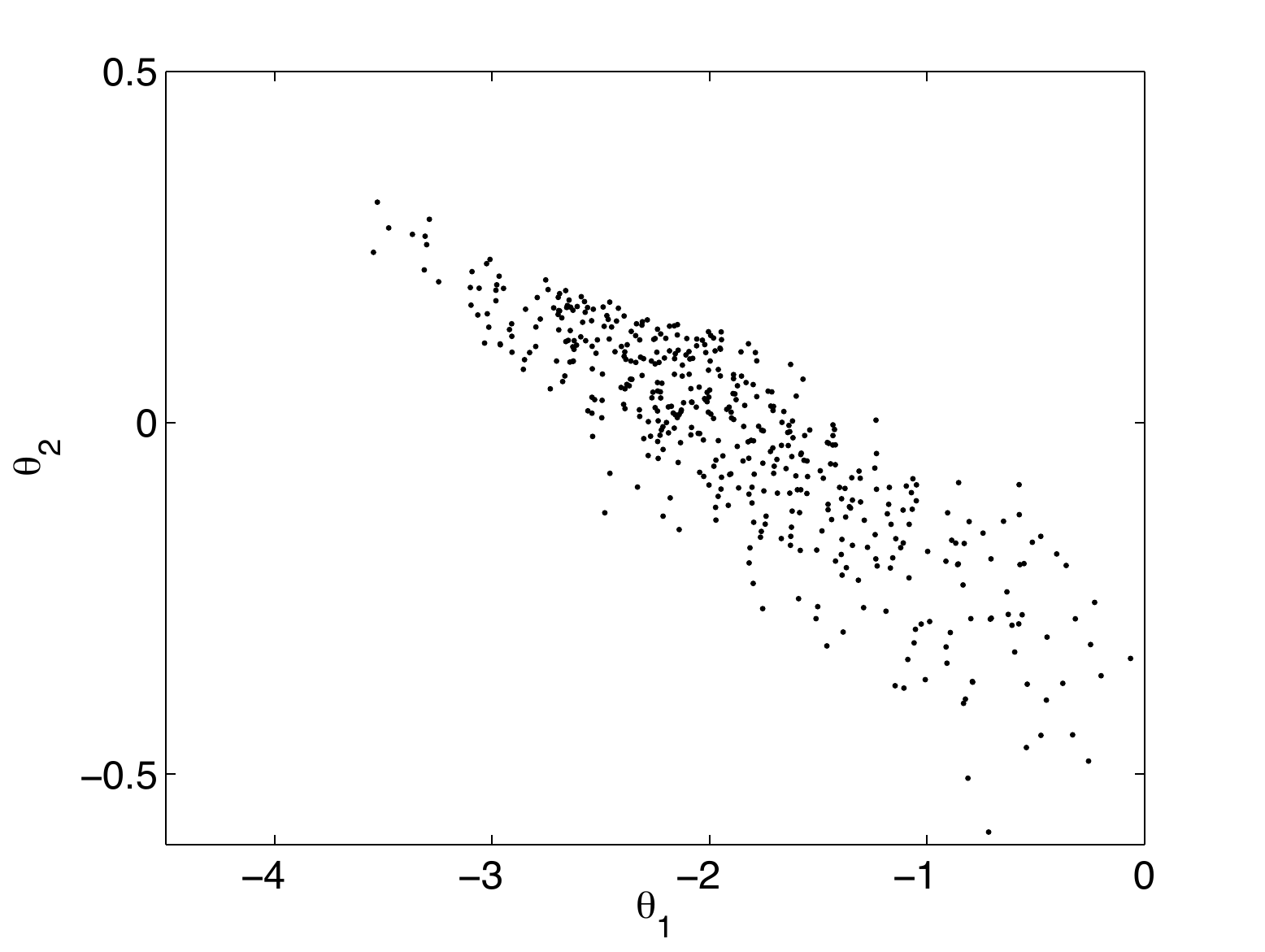}}
\end{figure}

\section{DISCUSSION\label{sec:Conclusions}}

Bayesian parameter estimation for latent MRFs can face computational
difficulties using standard methodology since the DA approach can
be extremely inefficient. We have described two methods, ABC and PMCMC
that can offer an alternative in situations where DA is not suitable.

ABC is currently particularly popular for addressing missing data
problems, and we apply it for the first time to inference in ERGMs
as a method for avoiding the intractable normalising
constant. We also provide a theoretical justification for the use
of MCMC for inexact simulation from the likelihood, as is required in many MRF models, within ABC-MCMC. However, we observe the
usual limitations of ABC in cases where sufficient statistics are
not available: namely that an approximation that is difficult to quantify
is introduced.

Marginal PMCMC offers an effective means of bypassing the potential
inefficiency of DA and our results indicate that this method is a promising approach to parameter estimation in these models.
We note the large computational cost of these methods, especially
when sampling from a high dimensional latent space. As such, routine
use of the method will usually only be possible when accompanied by
an efficient implementation, possibly on parallel computing architectures
such as graphics cards or cloud computing resources. Given the increase
in popularity of this hardware, the PMCMC methodology offers a promising
avenue for use on more realistic applications in the near future. The key aspect of PMCMC is the use of an SMC sampler for sampling
from the $x$ space. Although we have focussed on parameter estimation,
our results (and those of \citet{Hamze2005}) indicate that SMC samplers
have an important role to play in simulating from MRFs, both when
the parameters are known and unknown. These methods have not been
used before in the ERGM literature, and address some of the known
problems with MCMC described in \citet{Snijders2002}.

We follow previous work in using the exchange algorithm to account
for the intractable normalising constant, and give theoretical justification
for the use of approximate exchange algorithms where MCMC as a substitute
for exact simulation from the likelihood. In application, we find
the use of the extended version of the exchange algorithm described in \citet{Murray2006}
to be essential to ensure the MCMC can move freely in the problems
we consider.

\section*{Acknowledgments}

This work was funded by the EPSRC SuSTaIN program at the Department
of Mathematics, University of Bristol. The author thanks Christophe
Andrieu and Mark Briers for useful discussions and to the three anonymous reviewers whose comments helped to improve the paper.

\begin{center}
{\large\bf SUPPLEMENTAL MATERIALS}
\end{center}

\begin{description}

\item[Appendices] containing:  a derivation of the target distribution of exchange marginal PMCMC; and a derivation of bounds for the distance between the true posterior approximate posteriors targeted by the algorithms used in the paper with proof of an ergodicity result for the MCMC algorithms that target the approximate posteriors. (pdf)

\end{description}

\appendix

\section*{Appendix A: Target Distribution of Exchange Marginal PMCMC\label{sec:Target-Distribution-of}}

This appendix establishes that the EMPMCMC algorithm in section 2.2.3
of the paper has the desired target density of $p(\theta|y)$. The
proof requires the description of the extended target and proposal
distribution used as a consequence of the use of the SMC sampler within
the algorithm.

For ease of exposition, we express the algorithm in a slightly different
form to that in the main text, including the prior $p(\theta)$ within
the SMC sampler. The $i$'th iteration of the algorithm is then:
\begin{enumerate}
\item Draw $\theta^{*}\sim q(.|\theta(i-1))$ and $u\sim f(.|\theta^{*})$.
\item Run an SMC sampler on the $x$ space, with the final (unnormalised)
distribution as $\pi_{T}(x)=p(\theta)\gamma(x,y|\theta^{*})$ in order
to obtain the particle approximation $\widehat{p}(x|\theta^{*},y)$
to the distribution $\widehat{p}(x|\theta^{*},y)=p(\theta^{*})\gamma(x,y|\theta^{*})/Z(\theta^{*})p(\theta^{*},y)$
and an estimate $\widehat{\phi}(\theta^{*},y)$ of its normalising
constant, $\phi(\theta^{*},y):=Z(\theta^{*})p(\theta^{*},y)$.
\item Sample a single point $x^{*}$ from $\widehat{p}(.|\theta^{*},y)$.
\item Let $(\theta(i),\theta^{*}(i),x(i),\widehat{\phi}(i),u(i))=(\theta^{*},\theta(i-1),x^{*},\widehat{\phi}(\theta^{*},y),u)$
with probability:
\[
1\wedge \frac{\widehat{\phi}(\theta^{*},y)q(\theta|\theta^{*})\gamma(u|\theta)}{\phi(i-1)q(\theta^{*}|\theta)\gamma(u|\theta^{*})},
\]
otherwise set $(\theta(i),\theta^{*}(i),x(i),\widehat{\phi}(i),u(i))=(\theta(i-1),\theta^{*}(i-1),x(i-1),\widehat{\phi}(i-1),u(i-1))$.
\end{enumerate}
\selectlanguage{english}%
In advance of the proof we also make some definitions relating to
the use of the SMC sampler (using $\theta$ rather than $\theta^{*}$
throughout to simplify the notation).\foreignlanguage{british}{ We
define $\pi_{k}^{\theta}(x_{k})=\gamma_{k}^{\theta}(x_{k})/\phi_{k}(\theta)$
(for $k=1,...,T$) to be the sequence of targets used in the SMC sampler.
In our case, we take $\gamma_{T}^{\theta}(x_{T})=p(\theta)\gamma(x_{T},y|\theta)$
and $\phi_{T}(\theta)=\phi(\theta,y)$, so that $\pi_{T}^{\theta}(x_{k})=p(x|\theta,y)$.
The underlying construction of the SMC sampler is such that it targets
the artificially constructed sequence of distributions $\widetilde{\pi}_{k}^{\theta}(x_{1:k})=\widetilde{\gamma}_{k}^{\theta}(x_{1:k})/\phi_{k}(\theta)$
(for $k=1,...,T$), with
\[
\widetilde{\gamma}_{k}^{\theta}(x_{1:k})=\gamma_{k}^{\theta}(x_{k})\prod_{j=1}^{k-1}L_{j}^{\theta}(x_{j+1},x_{j}).
\]
The SMC sampler generates a weighted importance sample from each extended
target $\widetilde{\pi}_{k}^{\theta}(x_{1:k})$ in succession. We
denote the state of the $p$'th particle at the $k$'th target by
$x_{1:k}^{p}$ (the joint state of all particles at the $k$th target
is denoted by $x_{1:k}$). At the $k$'th target, the particles are:
resampled; moved using a transition kernel; then weighted.}

\selectlanguage{british}%
To describe the resampling step, we introduce the distribution $r(a_{k-1}|w_{k-1})$,
with $w_{k-1}$ denoting the weights of the particles at target $k-1$
and $a_{k-1}=(a_{k-1}^{1},...,a_{k-1}^{P})$ with $a_{k-1}^{p}$ giving
the index of the {}``parent'' particle of {}``child'' particle
$x_{1:k}^{p}$. This operation can be interpreted as the process by
which child particles at target $k$ choose their parent particles
from the population at target $k-1$. For the proof we also need to
define, for $k=1,...,T$ and $p=1,...,P$, the index $b_{k}^{p}$
which the ancestor particle of $x_{1:T}^{p}$ at target $k$ had at
that time.

Let $M_{1}(x_{1})$ be the initial proposal density and $M_{k}(x_{k-1},x_{k})$
for $k=2,...,T$ be the transition kernels used at each target. The
weighting step, applied to the $p$'th particle at the $k$'th target
after the resampling and move steps, finds the unnormalised weight
$\widetilde{w}_{k}^{p}$ of the particle at the current target:
\begin{equation}
\widetilde{w}_{k}^{p}=\frac{\gamma_{k}^{\theta}(x_{k})L_{k-1}^{\theta}(x_{k},x_{k-1}^{a_{k-1}^{p}})}{\gamma_{k-1}^{\theta}(x_{k-1}^{a_{k-1}^{p}})M_{k}(x_{k-1}^{a_{k-1}^{p}},x_{k})}.\label{eq:weights}
\end{equation}
The weight is then normalised:
\[
w_{k}^{p}=\frac{\widetilde{w}_{k}^{p}}{\sum_{p=1}^{P}\widetilde{w}_{k}^{p}}.
\]

An approximation to the target $p(x|\theta,y)$ is then given by
\[
\widehat{p}(dx|\theta,y)=\sum_{p=1}^{P}w_{T}^{p}\delta_{x_{T}^{p}}(dx)
\]
where $\delta$ is the Dirac delta function, with an estimate of its
normalising constant given by
\[
\widehat{\phi}(\theta,y)=\prod_{k=1}^{T}\left[\frac{1}{P}\sum_{p=1}^{P}\widetilde{w}_{T}^{p}\right].
\]

In advance of the theorem, we introduce the joint density of the variables
generated by the SMC algorithm that uses $P$ particles and $T$ targets,
defined on the space $\mathcal{X}^{TP}\times\left\{ 1,...,P\right\} ^{(T-1)P}$:
\begin{eqnarray}
\psi^{\theta}(x_{1},...,x_{T},a_{1},...,a_{T-1}) & := & \psi^{\theta}(x_{1})\prod_{k=2}^{T}\psi^{\theta}(a_{k-1}|x_{k-1})\psi^{\theta}(x_{k}|x_{k-1},a_{k-1})\nonumber \\
 & = & \left(\prod_{p=1}^{P}M_{1}^{\theta}\left(x_{1}^{p}\right)\right)\prod_{k=2}^{T}\left(r(a_{k-1}|w_{k-1})\prod_{p=1}^{P}M_{k}^{\theta}\left(x_{k-1}^{a_{k-1}^{p}},x_{k}^{p}\right)\right) \nonumber.
\end{eqnarray}

\begin{thm}
Let $O_{k}^{p}:=\sum_{j=1}^{P}\mathbb{I}(a_{k}^{j}=p)$ be the number
of offspring of particle $p$ at target $k$. If for any $p=1,...,P$
and $k=1,...,T$ the resampling scheme satisfies
\[
\mathbb{E}[O_{k}^{p}]=Pw_{k}^{p},
\]
(i.e., it is unbiased) and also that
\begin{equation}
r(a_{k}^{p}=m|w_{k})=w_{k}^{m},\label{eq:r_property}
\end{equation}
then the EMPMCMC algorithm is an MCMC algorithm
targeting a joint distribution that admits $p(\theta,x|y)$ as a marginal.\end{thm}
\begin{proof}
The structure of the proof is as follows. We first fully describe
the target and proposal densities used in the marginal PMCMC algorithm,
then combine these with the density of the $u$ variable generated
for the exchange algorithm, and show that the EMPMCMC algorithm performs
a deterministic {}``swap'' move on this extended target.

To begin, on the space $\mathcal{E}:=\Theta\times\mathcal{X}^{TP}\times\left\{ 1,...,P\right\} ^{(T-1)P+1}$,
we define the proposal and target for a marginal PMCMC algorithm.
To simply the notation, let $E=(\theta,v,x_{1},...,x_{T},a_{1},...,a_{T-1})$
and $E^{*}=(\theta^{*},v^{*},x_{1}^{*},...,x_{T}^{*},a_{1}^{*},...,a_{T-1}^{*})$.
The proposal is then
\begin{eqnarray*}
q^{N}(\theta,E^{*}) & := & q(\theta^{*}|\theta)w_{T}^{*v^{*}}\psi^{\theta^{*}}(x_{1}^{*},...,x_{T}^{*},a_{1}^{*},...,a_{T-1}^{*})
\end{eqnarray*}
where the weight $w_{T}^{*v^{*}}$ is present due to the sampling
of the index $v^{*}$ to generate $x^{*}$ from the particle approximation
$\widehat{p}(dx|\theta,y)$. The target density is given by

\begin{eqnarray*}
\tilde{\pi}^{N}(E) & := & \frac{p(\theta,x_{T}^{v}|y)\prod_{k=2}^{T}L_{k-1}^{\theta}(x_{k}^{b_{T}^{v}},x_{k-1}^{b_{T}^{v}})}{P^{T}}\frac{\psi^{\theta}(x_{1},...,x_{T},a_{1},...,a_{T-1})}{M_{1}^{\theta}\left(x_{1}^{b_{1}^{v}}\right)\prod_{k=2}^{T}r(b_{k-1}^{v}|w_{k-1})M_{k}^{\theta}\left(x_{k-1}^{b_{k-1}^{v}},x_{k}^{b_{k}^{v}}\right)},
\end{eqnarray*}
which has the desired $p(\theta,x|y)$ as a marginal.

Now to derive the acceptance probability of the EMPMCMC algorithm,
we express the algorithm in terms of a deterministic swap move on
the following extended target:
\[
\check{\pi}(E,E^{*},u):=\tilde{\pi}^{N}(E)q^{N}(\theta,E^{*})\gamma(u|\theta^{*})/Z(\theta^{*}).
\]
Specifically, at each iteration of the algorithm, we apply the transformation
$T:\mathcal{E}\times\mathcal{E}\times\mathcal{X}\rightarrow\mathcal{E}\times\mathcal{E}\times\mathcal{X}$
defined by
\[
T(E,E^{*},u)=(E^{*},E,u).
\]
The acceptance probability of this move is given by
\begin{eqnarray}
1\wedge\frac{\check{\pi}(T(E,E^{*},u))}{\check{\pi}(E^{*},E,u)} & = & 1\wedge\frac{\tilde{\pi}^{N}(E^{*})}{q^{N}(\theta,E^{*})}\frac{q^{N}(\theta^{*},E)}{\tilde{\pi}^{N}(E)}\frac{\gamma(u|\theta)Z(\theta^{*})}{\gamma(u|\theta^{*})Z(\theta)}.\label{eq:accept}
\end{eqnarray}
The term $\tilde{\pi}^{N}(E)/q^{N}(\theta^{*},E)$ simplifies as follows:
\begin{eqnarray}
\frac{\tilde{\pi}^{N}(E)}{q^{N}(\theta^{*},E)} & = & \frac{1}{P^{T}}\frac{p(\theta,x_{T}^{v}|y)\prod_{k=2}^{T}L_{k-1}^{\theta}(x_{k}^{b_{T}^{v}},x_{k-1}^{b_{T}^{v}})}{q(\theta|\theta^{*})w_{T}^{v}M_{1}^{\theta}\left(x_{1}^{b_{1}^{v}}\right)\prod_{k=2}^{T}r(b_{k-1}^{v}|w_{k-1})M_{k}^{\theta}\left(x_{k-1}^{b_{k-1}^{v}},x_{k}^{b_{k}^{v}}\right)}\nonumber \\
 & = & \frac{p(\theta,x_{T}^{v}|y)\prod_{k=2}^{T}L_{k-1}^{\theta}(x_{k}^{b_{T}^{v}},x_{k-1}^{b_{T}^{v}})}{q(\theta|\theta^{*})M\left(x_{1}^{b_{1}^{v}}\right)\prod_{k=2}^{T}M_{k}^{\theta}\left(x_{k-1}^{b_{k-1}^{v}},x_{k}^{b_{k}^{v}}\right)}\frac{\frac{1}{P^{T}}}{\prod_{k=1}^{T}w_{k}^{b_{k}^{v}}}\nonumber \\
 & = & \frac{p(\theta,x_{T}^{v}|y)\prod_{k=2}^{T}L_{k-1}^{\theta}(x_{k}^{b_{T}^{v}},x_{k-1}^{b_{T}^{v}})}{\phi(\theta,y)q(\theta|\theta^{*})M\left(x_{1}^{b_{1}^{v}}\right)\prod_{k=2}^{T}M_{k}^{\theta}\left(x_{k-1}^{b_{k-1}^{v}},x_{k}^{b_{k}^{v}}\right)}\frac{\frac{1}{P^{T}}\prod_{k=1}^{T}\sum_{p=1}^{P}\tilde{w}_{k}^{p}}{\prod_{k=1}^{T}\widetilde{w}_{k}^{b_{k}^{v}}}\nonumber \\
 & = & \frac{p(\theta)\gamma(x_{T}^{v},y|\theta)\prod_{k=2}^{T}L_{k-1}^{\theta}(x_{k}^{b_{T}^{v}},x_{k-1}^{b_{T}^{v}})}{q(\theta|\theta^{*})M\left(x_{1}^{b_{1}^{v}}\right)\frac{\gamma_{1}(x_{1}^{b_{1}^{v}})}{M_{1}(x_{1}^{b_{1}^{v}})}\prod_{k=2}^{T}M_{k}^{\theta}\left(x_{k-1}^{b_{k-1}^{v}},x_{k}^{b_{k}^{v}}\right)\frac{\gamma_{k}^{\theta}(x_{k}^{b_{k}^{v}})L_{k-1}^{\theta}(x_{k}^{b_{k}^{v}},x_{k-1}^{b_{k-1}^{v}})}{\gamma_{k-1}^{\theta}(x_{k-1}^{b_{k-1}^{v}})M_{k}(x_{k-1}^{b_{k-1}^{v}},x_{k}^{b_{k}^{v}})}}\frac{\widehat{\phi}(\theta,y)}{p(y)Z(\theta)}\nonumber \\
 & = & \frac{\widehat{\phi}(\theta,y)}{q(\theta|\theta^{*})p(y)Z(\theta)},\label{eq:ratio}
\end{eqnarray}
using equations (\ref{eq:weights}) and (\ref{eq:r_property}).

Now, substituting equation (\ref{eq:ratio}) into equation (\ref{eq:accept}),
we obtain
\begin{eqnarray*}
1\wedge\frac{\check{\pi}(T(E,E^{*},u))}{\check{\pi}(E^{*},E,u)} & = & 1\wedge\frac{\widehat{\phi}(\theta^{*},y)}{q(\theta^{*}|\theta)p(y)Z(\theta^{*})}\frac{q(\theta|\theta^{*})p(y)Z(\theta)}{\widehat{\phi}(\theta,y)}\frac{\gamma(u|\theta)Z(\theta^{*})}{\gamma(u|\theta^{*})Z(\theta)}\\
 & = & 1\wedge\frac{\widehat{\phi}(\theta^{*},y)}{\widehat{\phi}(\theta,y)}\frac{q(\theta|\theta^{*})}{q(\theta^{*}|\theta)}\frac{\gamma(u|\theta)}{\gamma(u|\theta^{*})},
\end{eqnarray*}
as required.
\end{proof}

\section*{Appendix B: Convergence of Approximate Algorithms\label{sec:Convergence-of-approximate}}

In the appendix we prove the convergence of MCMC algorithms that take
the following form. We note that the assumptions used for the proof
are relatively strong, and are not widely applicable. However, it
likely that similar (weaker) results exist under weaker assumptions:
the results in this paper are intended as the first steps towards
future work that would obtain results that hold more generally.

Suppose that we have an {}``exact'' MCMC algorithm, using transition
kernel
\begin{equation} \nonumber
K((\theta,u),(\theta^{*},u^{*}))=q(\theta^{*}|\theta)\pi_{\theta^{*}}(u^{*})\alpha((\theta,u),(\theta^{*},u^{*}))+\delta_{\theta,u}(\theta^{*},u^{*})r(\theta,u),
\end{equation}
where $\alpha((\theta,u),(\theta^{*},u^{*}))$ is such that $K$ is
an MCMC kernel that has an invariant distribution of $\pi(\theta,u)=\pi(\theta)\pi_{\theta}(u)$
and
\begin{equation} \nonumber
r(\theta,u)=1-\int_{\theta^{*},u^{*}}q(\theta^{*}|\theta)\pi_{\theta^{*}}(u^{*})\alpha((\theta,u),(\theta^{*},u^{*}))d\theta^{*}du^{*}.
\end{equation}
The SAV method takes this precise form. The theorem below characterises
the invariant distribution (where it exists), and the convergence
rate, of the {}``approximate'' MCMC algorithm given by the kernel
\begin{equation} \nonumber
\widetilde{K}_{M}((\theta,u),(\theta^{*},u^{*}))=q(\theta^{*}|\theta)L_{\theta^{*}}^{M}(v_{0},u^{*})\alpha((\theta,u),(\theta^{*},u^{*}))+\delta_{\theta,u}(\theta^{*},u^{*})\widetilde{r}(\theta,u),
\end{equation}
where $L_{\theta^{*}}^{M}(v_{0},u^{*})$ represents $M$ iterations
of an MCMC kernel with invariant distribution $\pi_{\theta^{*}}(u^{*})$,
beginning at an arbitrary fixed initial value $v_{0}\in\mathcal{X}$
and
\begin{equation} \nonumber
\widetilde{r}(\theta,u)=1-\int_{\theta^{*},u^{*}}q(\theta^{*}|\theta)L_{\theta^{*}}^{M}(v_{0},u^{*})\alpha((\theta,u),(\theta^{*},u^{*}))d\theta^{*}du^{*}.
\end{equation}

The same argument can be used the prove equivalent properties of the
approximate exchange and ABC-MCMC algorithms described in the main
text. In the exact versions of these algorithms the target distributions
and transition kernels have slightly different forms to that of the
SAV method:
\begin{itemize}
\item the exchange algorithm has the target distribution given in the main
text, and the proposal additionally contains the deterministic {}``swap''
move described in the main text;
\item the ABC-MCMC algorithm can be seen to target $ $$\pi(\theta)\pi_{\theta}(u)\pi_{\epsilon}(u|y)$
(changing the notation to be consistent with that used in the proof),
with the proposal taking the form $q(\theta^{*}|\theta)\pi_{\theta^{*}}(u^{*})$.
\end{itemize}
These differences also result in a different acceptance probability
to that used in the SAV algorithm, but have no impact on the structure
of the proof of the theorem.

Throughout the theorem and proof, $\left\Vert .\right\Vert $ represents
the total variation norm.
\begin{thm}
Suppose the Metropolis-Hastings kernel $K$ is uniformly ergodic,
i.e. there exists $C_{K}\in(0,\infty)$ and $\rho_{K}\in(0,1)$, such
that for any $(\theta_{0},u_{0})\in\Theta\times\mathcal{X}$ and $n\in\mathbb{N}$
\begin{equation} \nonumber
\left\Vert K^{n}((\theta_{0},u_{0}),.)-\pi(.)\right\Vert \leq C_{K}\rho_{K}^{n};
\end{equation}
$L_{\theta^{*}}$ is geometrically ergodic for all $\theta^{*}\in\Theta$,
i.e. for all $\theta^{*}\in\Theta$, there exists $C_{L}(v_{0},\theta^{*})\in(0,\infty)$
and $\rho_{L}\in(0,1)$, such that for $\pi_{\theta^{*}}$-a.e. $v_{0}\in\mathcal{X}$
and $n\in\mathbb{N}$
\begin{equation} \nonumber
\left\Vert L_{\theta^{*}}^{n}(v_{0},.)-\pi_{\theta^{*}}(.)\right\Vert \leq C_{L}(v_{0})\rho_{L}^{n};
\end{equation}
uniformly in $\theta^{*}$. Additionally, suppose that for some $D>0$,
$\sup_{\theta,\theta'}q(\theta'|\theta)\leq D$, where $q$ is the
proposal used in the kernel $K$, and for any $M\geq1$ there exists
a distribution $\widetilde{\pi}_{M}$ on $\Theta\times\mathcal{X}$
such that $\widetilde{\pi}_{M}\widetilde{K}_{M}=\widetilde{\pi}_{M}$.

Then for any $\epsilon\in(0,\rho_{K}^{-1}-1)$ there exists $M_{0}\in\mathbb{N}$,
$\widetilde{\rho}_{K}\in(\rho,\rho(1+\epsilon)]\subset(\rho,1)$ and
$\widetilde{C}\in(0,\infty)$ such that for all $M\geq M_{0}$, $(\theta_{0},u_{0})\in\Theta\times\mathcal{X}$
and $n\geq1$,
\begin{equation}
\left\Vert \widetilde{K}_{M}^{n}((\theta_{0},u_{0}),.)-\widetilde{\pi}_{M}(.)\right\Vert \leq\widetilde{C}\widetilde{\rho}_{K}^{n},\label{eq:approx_conv_rate}
\end{equation}
\begin{equation}
\left\Vert \pi(.)-\widetilde{\pi}_{M}(.)\right\Vert \leq C_{K}\frac{\epsilon}{1-\rho_{K}}.\label{eq:approx_target}
\end{equation}
\end{thm}
\begin{proof}
To begin, for any $\theta,u\in\Theta\times\mathcal{X}$ and $n\in\mathbb{N}$,
\begin{eqnarray}
\left\Vert \widetilde{K}_{M}^{n}((\theta,u),.)-K^{n}((\theta,u),.)\right\Vert  & \leq & \sum_{i=0}^{n-1}\left\Vert \widetilde{K}_{M}^{n-i-1}((\theta,u),(K-\widetilde{K}_{M})(K^{i}-\pi)(.))\right\Vert \nonumber\\
 & \leq & C_{K}\sup_{\theta,u}\left\Vert K((\theta,u),.)-\widetilde{K}_{M}((\theta,u),.)\right\Vert \sum_{i=0}^{n-1}\rho_{K}^{i} \nonumber\\
 & \leq & \frac{C_{K}}{1-\rho_{K}}\sup_{\theta,u}\left\Vert K((\theta,u),.)-\widetilde{K}_{M}((\theta,u),.)\right\Vert \label{eq:for_target}
\end{eqnarray}
We now bound the final term on the right hand side. We have:
\begin{eqnarray}
\sup_{\theta,u}\left\Vert K((\theta,u),.)-\widetilde{K}_{M}((\theta,u),.)\right\Vert  & = & \sup_{\theta,u}\sup_{fs.t.|f|\leq1}\left|\int K((\theta,u),(d\theta',du'))f(\theta',u')-\int\tilde{K}_{M}((\theta,u),(d\theta',du'))f(\theta',u')\right| \nonumber\\
 & = & \sup_{\theta,u}\sup_{fs.t.|f|\leq1}\left|\int\left[q(d\theta'|\theta)\pi_{\theta'}(du')\alpha((\theta,u),(d\theta',du'))+\delta_{\theta,u}(d\theta',du')r(\theta,u)\right]f(\theta',u')\right. \nonumber \\
 &  & \quad\left.-\int\left[q(d\theta'|\theta)L_{\theta'}^{M}(v_{0},du')\alpha((\theta,u),(d\theta',du'))+\delta_{\theta,u}(d\theta',du')\widetilde{r}(\theta,u)\right]f(\theta',u')\right| \nonumber\\
 & \leq & \sup_{\theta,u}\sup_{fs.t.|f|\leq1}\left|\int\left[q(d\theta'|\theta)\alpha((\theta,u),(d\theta',du'))f(\theta',u')\left\{ \pi_{\theta'}(du')-L_{\theta'}^{M}(v_{0},du')\right\} \right]\right| \nonumber \\
 &  & \qquad+\sup_{\theta,u}\sup_{fs.t.|f|\leq1}\left|f(\theta,u)\left\{ r(\theta,u)-\widetilde{r}(\theta,u)\right\} \right| \nonumber\\
 & \leq & \sup_{\theta,u}\sup_{fs.t.|f|\leq1}\int\left|q(d\theta'|\theta)\right|\left|\alpha((\theta,u),(d\theta',du'))\right|\left|f(\theta',u')\right|\left|\pi_{\theta'}(du')-L_{\theta'}^{M}(v_{0},du')\right| \nonumber \\
 &  & \qquad+\sup_{\theta,u}\sup_{fs.t.|f|\leq1}\left|f(\theta,u)\right|\left|r(\theta,u)-\widetilde{r}(\theta,u)\right| \nonumber\\
 & \leq & \int D\left|\pi_{\theta'}(du')-L_{\theta'}^{M}(v_{0},du')\right|+\sup_{\theta,u}\left|r(\theta,u)-\widetilde{r}(\theta,u)\right| \nonumber\\
 & \leq & DC_{L}(v_{0})\rho_{L}^{M}+\sup_{\theta,u}\left|r(\theta,u)-\widetilde{r}(\theta,u)\right| \nonumber
\end{eqnarray}
using the geometric ergodicity of $L_{\theta^{*}}$ for every $\theta^{*}\in\Theta$.
Using this property again, for the final term we have:
\begin{eqnarray}
\sup_{\theta,u}\left|r(\theta,u)-\widetilde{r}(\theta,u)\right| & = & \sup_{\theta,u}\left|\int_{\theta^{*},u^{*}}q(\theta^{*}|\theta)\alpha\left((\theta,u),(\theta^{*},u^{*})\right)\right. \nonumber \\
 &  & \qquad\left.\left(L_{\theta^{*}}^{M}(v_{0},u^{*})-\pi_{\theta^{*}}(u^{*})\right)d\theta^{*}du^{*}\right| \nonumber\\
 & \leq & C_{L}(v_{0})\rho_{L}^{M}\sup_{\theta,u}\left|\int_{\theta^{*},u^{*}}q(\theta^{*}|\theta)\alpha\left((\theta,u),(\theta^{*},u^{*})\right)d\theta^{*}du^{*}\right| \nonumber\\
 & \leq & C_{L}(v_{0})\rho_{L}^{M}. \nonumber
\end{eqnarray}
Combining these results we obtain
\begin{equation} \nonumber
\sup_{(\theta,u)\in\Theta\times\mathcal{X}}\left\Vert K((\theta,u),.)-\widetilde{K}_{M}((\theta,u),.)\right\Vert \leq(D+1)C_{L}(v_{0})\rho_{L}^{M}.
\end{equation}
Using this result, and the uniform ergodicity of $K$, we obtain that
for any $(\theta,u),(\vartheta,v)\in\Theta\times\mathcal{X}$
\begin{eqnarray}
\left\Vert \widetilde{K}_{M}^{n}((\theta,u),.)-\widetilde{K}_{M}^{n}((\vartheta,v),.)\right\Vert  & \leq & \left\Vert \widetilde{K}_{M}^{n}((\theta,u),.)-K_{M}^{n}((\theta,u),.)\right\Vert +\left\Vert K_{M}^{n}((\vartheta,v),.)-\widetilde{K}_{M}^{n}((\vartheta,v),.)\right\Vert \nonumber \\
 &  & \qquad+\left\Vert K_{M}^{n}((\theta,u),.)-K_{M}^{n}((\vartheta,v),.)\right\Vert \nonumber \\
 & \leq & 2\frac{C_{K}}{1-\rho_{K}}(D+1)C_{L}(u_{0})\rho_{L}^{M}+C_{K}\rho_{K}^{n}.\nonumber
\end{eqnarray}

Now define
\begin{equation} \nonumber
\widetilde{\rho}_{K}:=\rho_{K}\sqrt[n]{C_{K}\left(1+2\frac{\rho_{L}^{M}\rho_{K}^{-n}}{1-\rho}(D+1)C_{L}(v_{0})\right)}\leq\rho_{K}\sqrt[n]{C_{K}}\left(1+2\frac{\rho_{L}^{M}\rho_{K}^{-n}}{1-\rho}(D+1)C_{L}(v_{0})\right).
\end{equation}

Choose $\epsilon\in(0,\rho_{K}-1)$ and let $n\in\mathbb{N}$ be such
that \foreignlanguage{english}{$\sqrt[n]{C_{K}}\leq\sqrt{1+\epsilon}$
and $M_{1}\in\mathbb{N}$ such that
\begin{equation} \nonumber
1+2\frac{\rho_{L}^{M_{1}}\rho_{K}^{-n}}{1-\rho}(D+1)C_{L}(v_{0})\leq\sqrt{1+\epsilon}.
\end{equation}
}We then have that $\widetilde{\rho}_{K}\leq\rho_{K}(1+\epsilon)<1$
for $M>M_{1}$, and equation \ref{eq:approx_conv_rate} follows.

Then from equation \ref{eq:for_target} we have that for any $n\geq1$
\begin{eqnarray}
\left\Vert \pi\widetilde{K}_{M}^{n}((\theta,u),.)-\pi K^{n}((\theta,u),.)\right\Vert  & \leq & \sum_{i=0}^{n-1}\left\Vert \pi\left(\widetilde{K}_{M}^{n-i-1}((\theta,u),(K-\widetilde{K}_{M})(K^{i}-\pi)(.))\right)\right\Vert \nonumber \\
 & \leq & C_{K}(D+1)C_{L}(u_{0})\rho_{L}^{M}\sum_{i=0}^{n-1}\rho_{K}^{i} \nonumber \\
 & \leq & \frac{C_{K}}{1-\rho_{K}}(D+1)C_{L}(v_{0})\rho_{L}^{M}. \nonumber
\end{eqnarray}
So for any $\varepsilon>0$, we can chose $M_{2}\in\mathbb{N}$ such
that $(D+1)C_{L}(u_{0})\rho_{L}^{M_{2}}\leq\varepsilon$, so that
\begin{equation}\nonumber
\left\Vert \pi\widetilde{K}_{M}^{n}((\theta,u),.)-\pi K^{n}((\theta,u),.)\right\Vert \leq\frac{C_{K}\varepsilon}{(1-\rho_{K})}.
\end{equation}
 Using $\left\Vert \pi-\widetilde{\pi}_{M}\right\Vert =\lim_{n\rightarrow\infty}\left\Vert \pi K_{M}^{n}-\widetilde{\pi}_{M}\widetilde{K}_{M}^{n}\right\Vert $
and choosing $M_{0}=M_{1}\vee M_{2}$ the proof is completed.
\end{proof}
We note that the same argument may be used to obtain the same result
where $L_{\theta^{*}}$ is allowed to use the value of $u$ generated
at the previous iteration, as long as $L_{\theta^{*}}$ is uniformly
ergodic for every $\theta^{*}\in\Theta$.\selectlanguage{english}

\bibliographystyle{mychicago}
\addcontentsline{toc}{section}{\refname}\bibliography{refs}

\end{document}